\documentclass[a4paper,UKenglish,cleveref,autoref,thm-restate]{lipics-v2021}

\hideLIPIcs
\nolinenumbers

\usepackage[utf8]{inputenc}
\usepackage{multicol}
\usepackage[ruled,noend]{algorithm2e}
\usepackage[noend]{algorithmic} 
\usepackage{comment}
\usepackage{lineno}
\usepackage{tabularx}
\usepackage{makecell}
\usepackage{hyperref}
\usepackage{graphicx}
\usepackage{bm}
\usepackage{adjustbox}
\usepackage{csquotes}
\usepackage{thm-restate}

\usepackage[numbers,sort]{natbib}

\newtheorem{lem}{Lemma}

\newcommand{\eps}{\varepsilon}
\newcommand{\T}{\mathbb{T}}

\usepackage{float}
\usepackage{todonotes}

\newcommand{\R}{\mathbb{R}}
\newcommand{\N}{\mathbb{N}}

\newcommand{\set}[1]{\{#1\}}
\newcommand{\red}[1]{\textcolor{red}{#1}}
\DeclareMathOperator{\dist}{dist}
\DeclareMathOperator{\frechet}{D_G}
\DeclareMathOperator{\f}{D_T}
\DeclareMathOperator{\str}{str}
\DeclareMathOperator{\nnext}{next}

\title{Fréchet Distance in Unweighted Planar Graphs}

\author{Ivor van der Hoog}{Technical University of Denmark, Denmark}{idjva@dtu.dk}{
https://orcid.org/0009-0006-2624-0231}{}

\author{Thijs van der Horst} {Utrecht University, the Netherlands \and TU Eindhoven, the Netherlands} {t.w.j.vanderhorst@uu.nl} {https://orcid.org/0009-0002-6987-4489} {}

\author{Eva Rotenberg}{Technical University of Denmark}{erot@itu.dk}{0000-0001-5853-7909 }{}

\author{Lasse Wulf}{Technical University of Denmark, Denmark}{lawu@dtu.dk}{
https://orcid.org/0000-0001-7139-4092}{}

\authorrunning{I. van der Hoog, T. van der Horst E. Rotenberg, and L. Wulf}
\Copyright{I. van der Hoog, T. van der Horst E. Rotenberg, and L. Wulf}

\funding{{\it Ivor van der Hoog}, {\it Eva Rotenberg}, and {\it Lasse Wulf} are grateful to the Carlsberg Foundation for supporting this research via Eva Rotenberg's Young Researcher Fellowship CF21-0302 ``Graph Algorithms with Geometric Applications''. This work was supported by the the VILLUM Foundation grant (VIL37507) ``Efficient Recomputations for Changeful Problems'' and the European Union's Horizon 2020 research and innovation programme under the Marie Sk\l{}odowska-Curie grant agreement No 899987. }


\begin{document}

\maketitle

\begin{abstract}
The Fr\'echet distance is a distance measure between trajectories in $\mathbb{R}^d$ or walks in a graph $G$.  
Given constant-time shortest path queries, the Discrete Fr\'echet distance $\frechet(P, Q)$ between two walks $P$ and $Q$ can be computed in $O(|P| \cdot |Q|)$ time using a dynamic program. 
Driemel, van der Hoog, and Rotenberg [SoCG'22] show that for weighted planar graphs this approach is likely tight, as there can be no strongly-subquadratic algorithm to compute a $1.01$-approximation of  $\frechet (P, Q)$ unless  the Orthogonal Vector Hypothesis (OVH) fails.

Such quadratic-time conditional lower bounds are common to many Fréchet distance variants. However, they can be circumvented by assuming that the input comes from some well-behaved class: 
There exist $(1+\eps)$-approximations, both in weighted graphs and in $\mathbb{R}^d$, that take near-linear time for $c$-packed or $\kappa$-straight walks in the graph.
In $\mathbb{R}^d$ there also exists a near-linear time algorithm to compute the Fréchet distance whenever all input edges are long compared to the distance. 

We consider computing the Fréchet distance in unweighted planar graphs. 
We show that there exist no strongly-subquadratic $1.25$-approximations of the discrete Fr\'echet distance between two disjoint simple paths in an unweighted planar graph in strongly subquadratic time, unless OVH fails. 
This improves the previous lower bound, both in terms of generality and approximation factor.
We subsequently show that adding graph structure circumvents this lower bound: 
If the graph is a regular tiling with unit-weighted edges, then there exists an $\tilde{O}( (|P| + |Q|)^{1.5})$-time algorithm to compute $\frechet(P, Q)$. 
Our result has natural implications in the plane, as it allows us to define a new class of well-behaved curves that facilitate  $(1+\eps)$-approximations of their discrete Fréchet distance  in subquadratic time. 

\end{abstract}

\newpage
\section{Introduction}

The Fréchet distance is a widely used metric for measuring the similarity between trajectories. It is often illustrated through a metaphor: imagine a person walking along one trajectory and their dog following another, connected by a leash. The Fréchet distance is the minimum possible leash length over all synchronized traversals of both trajectories.

This metric has numerous applications, particularly in movement data analysis~\cite{buchin2017clustering, driemel2016clustering, kenefic2014track,chen2011approximate, bang2016improved, buchin2020group, konzack2017visual, xie2017distributed}. It is versatile, having been applied to handwriting recognition~\cite{sriraghavendra2007frechet}, coastlines~\cite{mascret2006coastline}, geometric shapes in geographic information systems~\cite{devogele2002new}, and trajectories of moving objects such as vehicles, animals, and athletes~\cite{acmsurvey20, su2020survey, brakatsoulas2005map, buchin2020group}. 
We consider the \emph{discrete} Fréchet distance, where trajectories are  modeled as sequences of discrete points.

Alt and Godau~\cite{alt1995computing} were the first to analyse the Fréchet distance from a computational perspective. They considered trajectories as polygonal curves in $\mathbb{R}^d$ with $n$ and $m$ vertices, and compute the continuous Fréchet distance in $O(mn \log (n + m))$ time. Later, Buchin, Buchin, Meulemans and Mulzer~\cite{buchin2017four} introduced a faster randomized algorithm, achieving a running time of $O(n^2 \sqrt{\log n} (\log \log n)^{3/2} )$ on a real-valued pointer machine and $O(n^2 \log \log n)$ on a word RAM.
For the discrete Fréchet distance, Eiter and Mannila~\cite{eitermannila94} presented an $O(nm)$ time algorithm under constant-time distance computations. Agarwal, Avraham, Kaplan, and Sharir~\cite{agarwal2014computing} later improved this to $O(nm (\log \log nm ) / \log nm )$ in the word RAM.

Driemel, van der Hoog, and Rotenberg~\cite{DBLP:conf/compgeom/DriemelHR22} 
initiate the study of Fréchet distance in graphs.
They consider as input a (weighted) graph $G$ where a trajectory is a walk in the graph.
The distance between two vertices is the length of their shortest path. 
Given a constant-time 
distance oracle, the discrete Fréchet distance $\frechet(P, Q)$ between any two walks in $G$ can then be computed in $O(nm)$ time using the algorithm by Eiter and Mannila~\cite{eitermannila94}.

\subparagraph{Conditional Lower Bounds.}
Several conditional lower bounds exist for computing the Fréchet distance or its constant-factor approximations. These results rely on complexity assumptions such as the Orthogonal Vector Hypothesis (OVH) and the Strong Exponential Time Hypothesis (SETH)~\cite{DBLP:journals/tcs/Williams05}.
Bringmann~\cite{DBLP:conf/focs/Bringmann14} established that no algorithm can compute the (discrete or continuous) Fréchet distance between two polygonal curves of $n$ vertices in $O(n^{2-\delta})$ time for any $\delta > 0$. The same lower bound holds for small constant-factor approximations. Bringmann’s proof originally involved self-intersecting curves in the plane, but later work by Bringmann and Mulzer~\cite{bringmann2016approximability} extended the result to intersecting curves in $\mathbb{R}^1$. Additionally, Bringmann~\cite{DBLP:conf/focs/Bringmann14} proved a lower bound for unbalanced inputs: given two curves with $n$ and $m$ vertices in the plane, no algorithm can compute the Fréchet distance in $O((nm)^{1-\delta})$ time assuming OVH.
Buchin, Ophelders, and Speckmann~\cite{buchin2019seth} further demonstrated that, assuming OVH, no $O((nm)^{1-\delta})$ time algorithm can compute better than a 3-approximation of the Fréchet distance for pairwise-disjoint planar curves in $\mathbb{R}^2$ or intersecting curves in $\mathbb{R}^1$.
Driemel, van der Hoog and Rotenberg~\cite{DBLP:conf/compgeom/DriemelHR22} give a similar lower bound for paths in a weighted planar graph.
They prove that, unless OVH fails, no strongly-subquadratic algorithm can $1.01$-approximate the discrete Fréchet distance between arbitrary paths in a weighted planar graph — even if the ratio between smallest and highest weight is bounded by a constant.
This lower bound applies only to weighted graphs and does therefore not exclude the existence of a fast algorithm on unweighted planar graphs.
A closer inspection of their argument reveals that their proof cannot simply be adapted to the unweighted case by subdividing long edges often enough.
In this case, new vertices get introduced to the graph and their arguments break down.

\subparagraph{Avoiding lower bounds through well-behaved curves.}
These lower bounds can be circumvented if the input is well-behaved. 
Driemel, Har-peled and Wenk~\cite{DriemelHW12} consider three classes of curves in $\mathbb{R}^d$ which are well-behaved as long as their corresponding behavioural parameter is low.
These are $c$-packed curves~\cite{DriemelHW12, bringmann17cpacked}, $\phi$-low density curves~\cite{DriemelHW12}, and $\kappa$-bounded curves~\cite{alt2004comparison,AronovHKWW06,DriemelHW12}.
They show algorithms to compute a $(1+\eps)$-approximation between two well-behaved curves whose running time is near-linear in $n$ and $m$.
 Gudmundsson, Mirzanezhad, Mohades, and Wenk consider the special case where all edges of the input curves are long with respect to their Fréchet distance~\cite{gudmundsson2019fast}. In this case, the Fréchet distance can be computed in $O( (n+m) \log (n+m))$ time~\cite{gudmundsson2019fast}.
 Driemel, van der Hoog, and Rotenberg~\cite{DBLP:conf/compgeom/DriemelHR22} show that if the input are walks in a graph $G$ and one of the walks is restricted to a $c$-packed or $\kappa$-straight path, then there exists $(1+\eps)$-approximations that take near-linear time. 

\subparagraph{Contribution.}
Driemel, van der Hoog, and Rotenberg~\cite{DBLP:conf/compgeom/DriemelHR22} ruled out a strongly subquadratic algorithm for a $1.01$-approximation between paths in a weighted graph, while the construction by Buchin, Ophelders, and Speckmann~\cite{buchin2019seth} rules out a strongly subquadratic algorithm for a $3$-approximation between walks in an unweighted planar graph.
We consider the discrete Fréchet distance between paths in an unweighted planar graph. We prove (Theorem~\ref{thm:main-lower-bound}) that no strongly-subquadratic $1.25$-approximation algorithm exists unless OVH fails, thereby improving both the generality (and also the approximation factor for planar graphs). 

Next, we prove that adding graph structure circumvents this lower bound. 
If $G$ is an unweighted regular tiling of the plane, then there exists an $\tilde{O}( (n+m)^{1.5})$-time algorithm to compute the Fréchet distance $\frechet(P, Q)$, between curves of length $n$ and $m$. 
Our results have natural implications in the plane, for a special class of well-behaved curves ($(\eps, \delta)$-curves, to be defined). For those, we can compute the Fréchet distance under the $L_1$ metric (respectively, approximate it for any $L_c$ metric) in $\tilde{O}(\frac{\sqrt{\delta}}{\eps} (n+m)^{1.5})$ time (respectively, $\tilde{O}(\frac{\sqrt{\delta}}{\eps \sqrt{\eps}} (n+m)^{1.5})$ time). 
For completeness, we explain the continuous analogy to $(\eps, \delta)$-curves in Appendix~\ref{app:continuous}.
We compare this new curve class to existing well-behaved curve classes in Appendix~\ref{app:realistic}.

\begin{figure}[b]
    \centering
    \includegraphics[page=1]{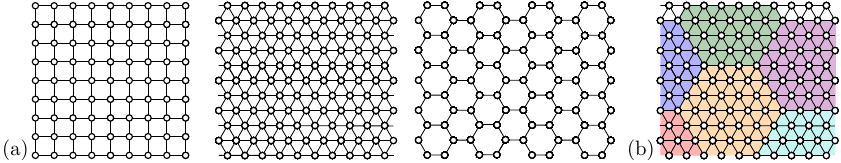}
    \caption{(a) The three regular plane tilings.  (b) A hexagonal super-tiling of a triangular tiling. }
    \label{fig:tilingdef}
\end{figure}

\section{Preliminaries}

Let $G$ be an unweighted graph where $V(G)$ is its vertex set.
We define $P = (p_1,\dots,p_n)$ and $Q = (q_1,\dots,q_m)$ as two paths in $G$. A path cannot visit a vertex twice. 
For integers $i, i' \in [n]$ with $i \leq i'$ we denote by $P[i: i']$ the subpath of $P$ from $p_i$ to $p_{i'}$. 
We define the discrete Fréchet distance $\frechet(P, Q)$ between $P$ and $Q$ as follows~\cite{DBLP:conf/compgeom/DriemelHR22}: 
A sequence $(a_i, b_i)_{i \in [k]}$ of pairs of indices is a \emph{monotone walk} if for all $i \in [k-1]$ we have $a_{i+1} \in \set{a_i, a_i+1}$ and $b_{i+1} \in \set{b_i, b_i + 1}$. Let $\mathbb{F}_{n, m}$ be the set of all monotone walks from $(1,1)$ to $(n,m)$. 

\newpage
\noindent
The cost of $F \in \mathbb{F}_{n, m}$ is the maximum of $\dist(p_i, q_j)$ over $(i,j) \in F$, where $\dist()$ denotes the length of the shortest path between vertices.  $\frechet(P, Q)$ is the minimum cost of a walk in $\mathbb{F}_{n, m}$.
\[
\frechet(P,Q) = \min_{F \in \mathbb{F}_{n, m}} \text{cost}(F) = \min_{F \in \mathbb{F}_{n, m}} \max_{(i,j) \in F} \dist(p_i, q_j).
\]

\subparagraph{Decision variant.}
We focus on the decision variant where the input includes some  $\Delta \geq 0$ and the goal is to output whether $\frechet(P, Q) \leq \Delta$. 
The \emph{$\Delta$-free space matrix} $M_\Delta$ is the $n$ by $m$ matrix where $M_\Delta[i, j]$ is $0$ if $\dist(p_i, q_j) \leq \Delta$ and $1$ otherwise.
We follow geometric convention, where $i$ denotes the column of $M_\Delta$ (column $i$ corresponds to some $p_i \in P$. Column entries with value $0$  correspond to $q_j \in Q$ with $\dist(p_i, q_j) \leq \Delta$).
Eiter and Mannila~\cite{eitermannila94} define a directed graph over $M_\Delta$ where there is an edge from $M_\Delta[a, b]$ to $M_\Delta[c, d]$ if and only if $c \in \{a, a+1 \}$, $d \in \{b,b+1 \}$ and $M_\Delta[a, b] = M_\Delta[c, d] = 0$. 
They prove that $\frechet(P, Q) \leq \Delta$ if and only if there exists a directed path from $(1, 1)$ to $(n, m)$ in this graph.

\subparagraph{Orthogonal vectors.} Our conditional lower bound reduces from the Orthogonal Vector Hypothesis. The input consists of two ordered sets of $d$-dimensional Boolean vectors $U,W$ of sizes $|U| = n$ and $|W| = m$. We denote for $u \in U$ by $u_i \in \{0,1\}$ its $i$'th component. 
The \emph{OV problem} is to find out if some vector of $U$ is orthogonal to some vector of $W$.

\begin{definition}[Williams \cite{DBLP:journals/tcs/Williams05}] The orthogonal vector hypothesis (OVH) states that for all $\delta > 0$, there exists constants $\omega$ and $1 > \gamma > 0$ such that the OV problem for vectors of dimension $d = \omega \log n$ and $m = n^\gamma$ cannot be solved in $O((nm)^{1-\delta})$ time. 
\end{definition}

We consider each vector in $U$ (or $W$) as a string of $d$ bits. We denote by $\str(U)$ and $\str(W)$ the concatenation of all strings in $U$ and $W$, respectively. 
A \emph{substring} of some string $s_1 s_2 \, \dots \, s_k$ is the sequence $s_a s_{a+1} \, \dots \, s_{b-1} s_b$ for some $1 \leq a \leq b \leq k$.

\subparagraph{Regular tilings.}
A \emph{regular tiling} $T$ is an infinite plane graph where each face is a regular polygon and each face has the same degree. 
There exist exactly three regular tilings~\cite{grunbaum1987tilings}: the triangular, square, and hexagonal tiling (Figure~\ref{fig:tilingdef}). Without loss of generality, we assume that our tilings are axis-aligned. 
A face $F$ of $T$ is an open region bounded by the edges of $T$, we denote by $\overline F$ its closure.
The \emph{edge length} of $T$ is the length of any edge. A tiling is \emph{unit} if its edges have unit length.
A path in $T$ is a path in its corresponding graph.  
We always denote by $T$ a unit tiling where $(0, 0)$ is a vertex of $T$. 
%For any face $F$ of a tiling, we denote by $\partial{F} = \overline{F}-F$ its boundary.
For any pair of vertices $p, q$ of $T$, the term $\dist(p, q)$ denotes the distance in the graph $T$. We assume that we can compute $\dist(p, q)$ in constant time.
Our algorithm considers what we will call a \emph{super-tiling} of $T$ and a natural partition of the edges of $P$ and $Q$ into subpaths. 

\begin{definition}[Figure~\ref{fig:tilingdef} (b)]
    Given a regular tiling $T$ of the plane, 
    a \emph{super-tiling} $\T$ of $T$ is a regular tiling of the plane whose vertex set is a subset of the vertex set $V(T)$. 
    The \emph{boundary vertices} $B(\T)$ are the vertices of $T$ that lie on the boundary of some face of $\T$. 
    %We will write $B(\T)$ when $T$ is apparent from context. 
\end{definition}

%\begin{definition}[Figure~\ref{fig:tilingdef} (a)]
 %   Given a regular tiling $T$ of the plane, 
    %an integer $f$, 
 %   and a vertex $s \in T$, 
 %   a \emph{super-tiling} $T'$ is a regular tiling of the plane where:
 %   \begin{itemize}
 %       \item all vertices of $T'$ coincide with vertices of $T$, 
 %       \item the tiling $T_f^s$ contains the vertex $s$ of $T$, and
%        \item edges of $T_f^s$ are straight lines of length $f$.
%    \end{itemize}
%    For any face $F$ of $T_f^s$ we denote  
%    by $V(F)$ all tile vertices that coincide with its boundary $\widehat{F}$, and by $B_f^s(T) := \bigcup\limits_{\textnormal{Face } F \textnormal{ of } T_f^s } V(F)$.
%\end{definition}

%Note that for a regular tiling, for any integer $f$, there exists 

\begin{definition}
    Consider a regular tiling $T$, some super-tiling $\T$ and a path $P$ in $T$, we define the \emph{breakpoints} $B(P, \T)$ as all vertices in $P \cap B(\T)$.    
\end{definition}

Given a path $P$ in a tiling $T$ and a super-tiling $\mathbb{T}$ of $T$, we want to consider how $\mathbb{T}$ splits $P$ into smaller snippets, each snippet ``belonging'' to a face of $\mathbb{T}$. More formally:

\begin{definition}[Figure~\ref{fig:tilingdef2}]
\label{def:induced}
    Consider a regular tiling $T$ and a super-tiling $\T$ of $T$.
    Given a path $P= (p_1,p_2,\ldots,p_n)$ in $T$, let $\mathbb{I} := \{ i \in [n] \mid p_i \in  B(\mathbb{T}) \textnormal{ or } i \in \{1, n \} \}$.
    We define the \emph{induced subpaths} $S(P, \T)$ as all subpaths $P[x:y]$, for consecutive $x,y\in\mathbb{I}$.
Note that for each induced subpath, there exists a (not necessarily unique) face $F$ in $\T$ whose closure $\bar{F}$ contains it; we say the induced subpath is \emph{assigned} to one such face (breaking ties arbitrarily).
\end{definition}

\begin{figure}[b]
    \centering
    \includegraphics[page=2]{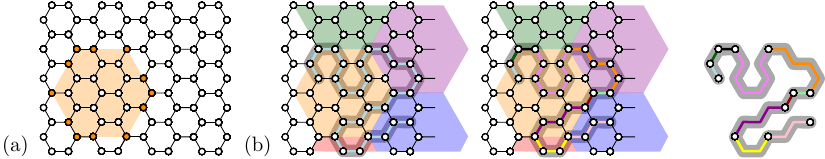}
    \caption{(a) For a face $F$ in the super-tiling $\T$, we show the set $B(\T) \cap \overline F$. (b) Given a path $P$, we create a sequence of subpaths of $P$ where consecutive subpaths overlap in a single vertex.  }
    \label{fig:tilingdef2}
\end{figure}

\subparagraph{Submatrices of $M_\Delta$.}
We show a subquadratic algorithm to decide whether 
 $\f(P, Q) \leq \Delta$ if $P$ and $Q$ are paths in a tiling $T$. We compute a monotone walk from $(1, 1)$ to $(n, m)$ in $M_\Delta$ by efficiently propagating reachability information across \emph{submatrices} of $M_\Delta$.
For integers $i, i' \in [n]$ and $j, j' \in [m]$ with $i \leq i'$ and $j \leq j'$ we denote by $M_\Delta[i :i', j :j']$ the submatrix of $M_\Delta$ that consist of all entries $M_\Delta[x, y]$ for $(x, y) \in [i, i'] \times [j, j']$. 
We define the \emph{bottom-left} facets of $M_\Delta[i :i', j :j']$ as all entries $\{ M_\Delta[i, y] \mid y \in [j, j'] \}$ and $\{ M_\Delta[x, j] \mid x \in [i, i']\}$. 
The \emph{top-right} facets are defined as the entries $\{ M_\Delta[i', y] \mid y \in [j, j'] \}$ and $\{ M_\Delta[x, j'] \mid x \in [i, i']\}$.

\subparagraph{$(\eps, \delta)$-curves.}
Finally, we introduce a new class of well-behaved curves in the plane. Intuitively, an $(\eps, \delta)$-curve, for small $\delta$, is a curve composed of short edges that does not revisit the same region of the plane too frequently.
Formally, fix a universal constant $\gamma \in O(1)$. Let $P$ be a curve in the plane, and let $B_\eps$ denote the set of all balls (in Euclidean metric) of radius $\eps$ centred at the vertices of $P$. We say that $P$ is an $(\eps, \delta)$-curve if all its edges have length at most $\gamma$, and any point in the plane lies in at most $\delta$ balls from $B_\eps$.

Note that for any given pair $(P, \eps)$, the value of $\delta$ is determined. As $\eps$ decreases, so may $\delta$. If $\delta$ remains large for small values of $\eps$, then  many vertices of $P$ are densely clustered— i.e., effectively, the curve frequently revisits the same locations.
We present an exact algorithm for the Discrete Fréchet distance under the $L_1$ metric with a linear dependency on $\frac{\sqrt{\delta}}{\eps}$. We assume that the input specifies a desirable value of $\eps$, and that this yields a favourable corresponding $\delta$.
Imposing that $\delta$ is constant amounts to requiring that the curve does not revisit any $\eps$-ball more than a constant number of times --
a restriction that somewhat resembles the well-studied $c$-packed curves~\cite{DriemelHW12, DBLP:conf/compgeom/DriemelHR22, bringmann17cpacked}. 
%There are subtle differences, however.
%E.g., $c$-packed curves allow for arbitrary edge lengths but impose restrictions on any ball, of any radius.
In Appendix~\ref{app:realistic} we discuss in detail how our new class relates to the prior well-behaved curve classes. 
%The restriction, $\gamma$, on edge-lengths corresponds for having to ``pay'' for long edges by subdividing them.

\section{A lower bound for paths in an unweighted planar graph}

Let $G$ be an unweighted planar graph and $P$ and $Q$ be paths in $G$. 
We show a lower bound, conditioned on OVH, that excludes any strongly subquadratic algorithm to compute a $1.25$-approximation of $\frechet(P, Q)$. 
Since we are space-restricted, we present our argument only in Appendix~\ref{app:lower_bound}. Here, we assume that the reader is familiar with how these lower bounds are typically constructed~\cite{DBLP:conf/focs/Bringmann14, buchin2019seth, DBLP:conf/compgeom/DriemelHR22, bringmann2016approximability} and only mention how our construction differs from prior works. 
We start with some OV instance $(U, W)$. 
We build an unweighted planar graph $G$ and two paths $P$ and $Q$ in $G$ such that $\frechet(P,Q) \leq 4$ if and only if the OV instance is a yes-instance and $\frechet(P,Q) \geq 5$ otherwise. 
One novelty of our approach is a preprocessing step, where we transform $(U, W)$ into an equivalent instance with additional properties:

\begin{restatable}{lem}{preprocessing}
\label{lem:preprocessing}
An instance $U',W' \subseteq \{0, 1\}^{d'}$ of OV can be preprocessed in $O(d'(n+m))$ time, resulting in a new instance $U, W \subseteq \{0, 1\}^{d}$ with $d \in O(d')$ such that:
\begin{itemize}
    \item a yes-instance stays a yes-instance and a no-instance stays a no-instance,
    \item for all $u \in U$, $u_1 = u_{d} = 0$ and for all $w \in W$, $w_1 = w_{d} = 1$,
    \item if the instance is a no-instance, then $ \forall u \in U$ the vector $u$ is not only non-orthogonal to every $w \in W$, but even non-orthogonal to all consecutive length-$d$ substrings of $\str(W)$.
\end{itemize}
\end{restatable}

\vspace{0.2cm} \noindent
This preprocessing significantly simplifies our proofs and may be of independent interest. 
Previous lower bound constructions~\cite{DBLP:conf/focs/Bringmann14, DBLP:conf/compgeom/DriemelHR22, buchin2019seth, bringmann2016approximability} fix some $\Delta \in \mathbb{R}$. Given an OV-instance $(U, W)$, they construct curves $(P, Q)$ by constructing for each $u \in U$ a subcurve $f(u)$ and each $w \in W$ a subcurve $g(w)$ (see Figure~\ref{fig:gadget_comparison}). 
The construction has a very strong property: 
Consider a traversal of $P$ and $Q$ where the person and the dog remain within distance $\Delta$ of one another. If the person is in $f(u)$ for some $u \in U$ and the dog is in $g(w)$ for some $w \in W$ then neither the person or the dog can stand still if the other moves.

Requiring the planar graph to be unweighted significantly restricts the freedom we have in engineering the pairwise distance between vertices. We are unable to recreate an equally strong property and instead obtain a much weaker statement: If the person is traversing the subpath $f(u)$ for some $u \in U$ and the dog is traversing $g(w)$ for some $w \in W$ and the person moves at least two steps then the dog has to move at least one. We rely upon our preprocessing to then still argue that there exists a traversal of $P$ and $Q$ such that the leash length is at most $4$ if and only if there is an orthogonal pair $(u, w) \in U \times W$.  

\begin{restatable}{theorem}{lowerbound}
\label{thm:main-lower-bound}
    Let $P$ and $Q$ be paths in an unweighted planar graph. Let $|P| = n$ and $|Q| = m = n^\gamma$ for some constant $\gamma > 0$. Then for all $\delta \in (0,1)$, one cannot approximate $\frechet(P, Q)$ by a factor better than $1.25$ in $O((nm)^{1-\delta})$ time unless OVH fails.
\end{restatable}
%----------------------
\begin{figure}[H]
    \centering
    \includegraphics[width = \linewidth]{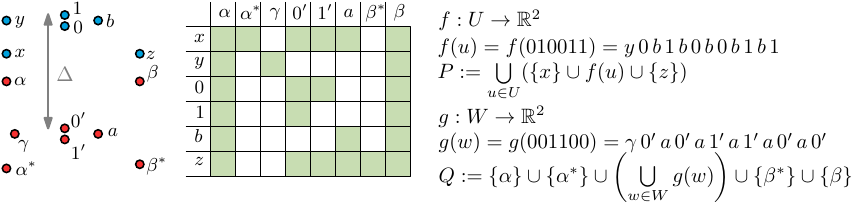}
    \caption{We illustrate the construction for planar curves by Bringmann~\cite{DBLP:conf/focs/Bringmann14}. 
    Each $u \in U$ becomes a subcurve $f(u)$ of $P$ and each $w \in W$ becomes a subcurve $g(w)$ of $Q$. Our table shows a green cell if the distance between the vertices is less than $\Delta$. If the person is in a subcurve $f(u)$ and the dog is in $g(w)$ and the person moves from its vertex then the dog \emph{must} also move.   
    }
    \label{fig:gadget_comparison}
\end{figure}

\section{Upper bound}
\label{sec:upper_bound}

We show that subquadratic algorithms \emph{do} exist for unweighted planar graphs that exhibit additional structure. Specifically, our main result is as follows:

\begin{restatable}{theorem}{upperbound}
    \label{thm:upperbound}
   Let $T$ be a regular unit tiling and $(P, Q)$ be paths in $T$ with \mbox{$|P| = n$ and $|Q|$} $= m$.
    Given $\Delta \in \mathbb{R}$, we can output whether $\f (P, Q) \leq \Delta$ in $O((n+m)^{1.5} \log (n+m))$ time.  
\end{restatable}

\vspace{1em}

We compute $\f (P, Q)$ as follows.
Given $(P, Q)$, we compute the smallest axis-aligned rectangles $R_P$ and $R_Q$ that contain $P$ and $Q$, respectively, in $O(n + m)$ time. Let $d_0$ be the Euclidean length of the shortest segment $(p, q)$ with $p \in R_P$ and $q \in R_Q$, and let $d_1$ be the Euclidean length of the longest such segment. 
Observe that $\f (P, Q) \in [\frac{d_0}{2}, 2d_1]$,
that this range contains $O(n + m)$ integers, and that $\f (P, Q)$ is always integral.
By performing a binary search over $[\frac{d_0}{2}, 2d_1] \cap \mathbb{Z}$, applying Theorem~\ref{thm:upperbound} at each step, we obtain:

\begin{corollary}
    Let $T$ be a regular unit tiling, and let $P$ and $Q$ be paths in $T$ with $|P| = n$ and $|Q| = m$.
    Then $\f (P, Q)$ can be computed in $O((n + m)^{1.5} \log^2 (n + m))$ time.
\end{corollary}

\subparagraph{Additional definitions. }
To prove Theorem~\ref{thm:upperbound}, we first introduce a few new concepts. Let $T$ be a regular, unit, axis-aligned tiling.
Based on $T$, we construct a super-tiling $\T$ of the same type as $T$, that is, if $T$ is a square/triangular/hexagonal tiling, then $\T$ is also a square/triangular/hexagonal tiling.
We state our results in general terms, but for ease of reading, the reader may assume $T$ is a square tiling.
For a (super) tiling $\T$, we define the \emph{halfslabs} of any face $F$ of $\T$ by picture (Figure~\ref{fig:halfslabs}). 
%-----------------------------------------------------------
\begin{figure}[b]
    \centering
    \includegraphics[page=2,height=0.23\textwidth]{img/halfslabs-new.pdf}
    \hspace{0.5cm}
    \includegraphics[page=3,height=0.25\textwidth]{img/halfslabs-new.pdf}
    \hspace{0.5cm}
    \includegraphics[page=1,height=0.25\textwidth]{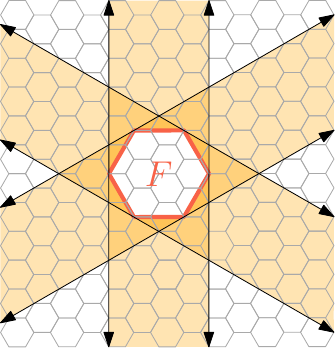}
    \caption{For a face $F$ in a square/triangular/hexagonal (super-)tiling, its halfslabs are the infinite slabs enclosing $F$ along the directions depicted here (orange).}
    \label{fig:halfslabs}
\end{figure}

\begin{definition}
    Let $T$ be a regular tiling, and let $\T$ be a super-tiling of $T$.
    We say that two faces $F$ and $F'$ of $\T$ are \emph{well-separated} by a vertex $v \in V(T)$ if, for all $p \in V(T) \cap \overline{F}$ and $q \in V(T) \cap \overline{F'}$, there exists a shortest path from $p$ to $q$ in $T$ that passes through $v$.
\end{definition}

\noindent
The knowledge that a certain pair $F,F'$ of faces is well-separated proves to be very useful, 
since in such a pair the distance between some vertex $p \in V(T) \cap \overline F$ and $q \in V(T) \cap \overline{F'}$ can be well-understood.
We provide an efficient algorithm to compute the pairwise Fréchet distance between subpaths of $P$ and $Q$ whenever their corresponding faces in $\T$ are well-separated. 
We identify a sufficient condition that ensures that some pair $F,F'$ is well-separated, based on the \emph{alignment} of the two:

\begin{definition}
    Let $T$ be a regular tiling.
    Two faces $F$ and $F'$ of $T$ are \emph{aligned} if $F'$ intersects a halfslab of $F$ (or $F$ intersects a halfslab of $F'$). 
\end{definition}

\begin{lemma}
    \label{lemm:aligned}
Let $T$ be an axis-aligned regular unit tiling, and let $\T$ be a super-tiling of $T$ of the same type as $T$.
Any two faces $F, F'$ of $\T$ that are not aligned are well-separated by a vertex $v \in V(T)$.
\end{lemma}

\begin{proof}
Consider first the case where $T$ is a square tiling (and so $\T$ is too).  The situation is depicted in Figure~\ref{fig:well-separated-square}~(a).
If the faces $F,F'$ are not aligned, then $F'$ is contained in one of the four regions $(R_1, R_2, R_3, R_4)$ of the plane formed by the halfslabs of $F$.
Let $v$ be the vertex at the apex of the region which contains $F'$. 
We claim that for any vertices $p \in V(T) \cap \overline F$ and $q \in V(T) \cap \overline{ F'}$, there exists a shortest path between $p$ and $q$ that passes through $v$. 
Hence $F$ and $F'$ are well-separated.
The claim is easy to see from the properties of shortest paths in grid graphs. 
We provide an alternative proof of the claim, with the advantage that it can be more easily generalized to the triangular and hexagonal grid.
The proof is depicted in Figure~\ref{fig:well-separated-square}~(b). Consider for $i =1,2,\dots$ the circles of equal radius in the plane graph $T$ centred around $p$, and $q$.
\[
h_i := \set{x \in V(T) : \dist(p,x) = i}, \quad h'_i := \set{x \in V(T) : \dist(q,x) = i}.
\]
We can see that $h_i, h'_i$ can be described as diamonds with four sides. 
Let us among the four sides of $h_i$ denote by $X_i$ the side that is facing towards $F'$. 
Likewise, let us among the four sides of $h'_i$ denote by $X'_i$ the side that is facing towards $F$.
Note that all of $X_i, X'_i$ are parallel for $i \in \N$. 
Let $\ell$ denote the line through vertex $v$ parallel to the $X_i$. 
Let $a \in \N$ be minimal such that $h_a$ touches $\ell$. 
Since $X_a$ is parallel to $\ell$, we have $X_a \subseteq \ell$. 
Furthermore, no matter which $p \in V(T) \cap \overline F$ was chosen as the centre of the circles $h_i$, we always have $v \in X_a$. This follows due to the way and the angle in which the circles $h_i$ expand. 
Analogously, let $b \in \N$ be minimal such that $h'_b$ touches $\ell$. Due to the waSince $F$ and $F'$ are not aligned, we have $v \in X'_b$. 
Now, the circle $h_a$ is entirely on one side of $\ell$, the circle $h'_b$ is entirely on the other side, and $v \in h_a \cap h'_b$. 
By definition of $h_a, h'_b$, this implies $\dist(p,q) = a+b$ and there is a shortest path between $p,q$ that passes through $v$.

For the case of triangular and hexagonal grids, the situation is depicted in Figures~\ref{fig:well-separated-triangle} and \ref{fig:well-separated-hexagon}.
Here the halfslabs divide the plane into 6 cone-shaped regions $R_1,\dots,R_6$. Let $v_j$ be the point at the apex of $R_j$ for $j =1,\dots,6$. Note that $v_j$ is also a vertex of $V(T)$.
For both the triangular and hexagonal grid, the definitions of $h_i,h'_i,X_i,X'_i,\ell,a,b$ can be made in an analogous matter. 
Analogously, it follows that $v_j \in h_a \cap h'_b$ which concludes the lemma.
\end{proof}
\begin{figure}[htpb]
    \centering
    \includegraphics[page=3]{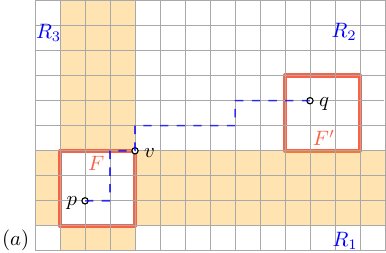}
    \includegraphics[page=4]{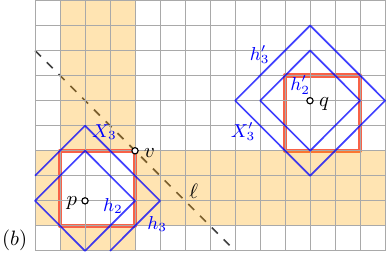}
    \caption{(a) The two faces $F,F'$ of the square super-tiling $\T$ are non-aligned and therefore well-separated by vertex $v$. (b) Proof of that fact.
    }
    \label{fig:well-separated-square}
\end{figure}
\begin{figure}[htpb]
    \centering
    \includegraphics[page=5]{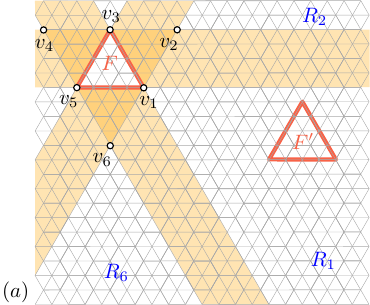}
    \includegraphics[page=6]{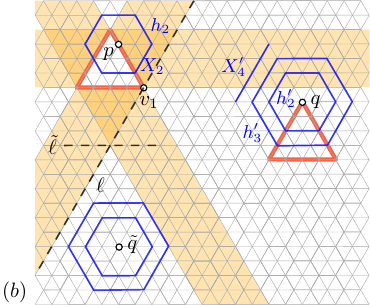}
    \caption{(a) Vertex $v_1$ well-separates non-aligned faces in a triangular grid. (b) Proof.
    }
    \label{fig:well-separated-triangle}
\end{figure}
\begin{figure}[htpb]
    \centering
    \includegraphics[page=1]{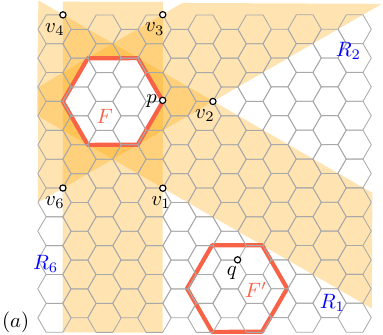}
    \includegraphics[page=2]{img/well-seperated.pdf}
    \caption{(a) Vertex $v_1$ well-separates non-aligned faces in a hexagonal grid. (b) Proof.
    }
    \label{fig:well-separated-hexagon}
\end{figure}

\newpage

\begin{figure}[b]
    \centering
    \includegraphics[]{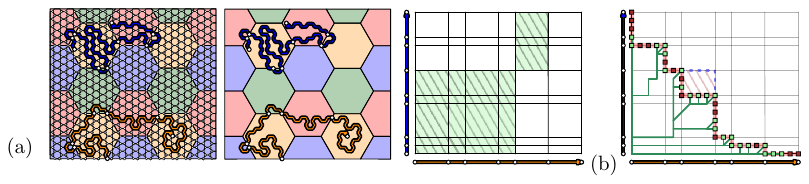}
    \caption{(a) For a given tiling $T$, we first compute a super-tiling $\T$ which partitions the edges of $P$ and $Q$. These subpaths divide the $\Delta$-free space matrix $M_\Delta$ into submatrices that correspond to aligned and non-aligned curve pairs (green dashed versus neutral).
    (b) We maintain a wave-front $W$ which is a monotone curve through $M_\Delta$ that has a Boolean value at each entry. An \emph{update} selects a submatrix $M_\Delta[i:i', j:j']$ and removes from $W$ its bottom-left facets (adding its top-right facets).    
    }
    \label{fig:alg_overview}
\end{figure}

\subparagraph{Algorithm overview.} 
Given $T$ and paths $P$ and $Q$, we first compute a suitable super-tiling:

\begin{lemma}
    \label{lem:properties}
 Let $T$ be a regular unit tiling, and let $P$ and $Q$ be two paths in $T$ of lengths $n$ and $m$. 
    In $O((n + m)^{1.5})$ time, we can compute a super-tiling $\T$ such that:
    \begin{itemize}
        \item the edges of $\T$ have a length in $\Theta(\sqrt{n + m})$, and
        \item the total number of induced subpaths in $S(P, \T) \cup S(Q, \T)$ (see Definition~\ref{def:induced}) is $O(\sqrt{n + m})$.
    \end{itemize}
\end{lemma}
\begin{proof}
    Our proof is inspired by that of Chan and Rahmati~\cite[Lemma~1]{chan2018improved}, who prove an analogous statement for grids and curves in higher-dimensional real space (without graphs).
    For the proof, we assume without loss of generality that $T$ is axis-aligned.

    Let $f = \left\lceil \sqrt{n + m} \ \right\rceil$, and let $\T_0$ be any axis-aligned super-tiling of $T$ with an edge length in $\Theta(f)$.
    For each $i \in [f - 1]$, define $\T_i$ to be the super-tiling obtained by shifting $\T_0$ by the vector $(i, i)$ if $T$ is square, or by $(0, i \sqrt{3})$ if $T$ is triangular or hexagonal. These shifted super-tilings remain axis-aligned super-tilings of $T$, as all their vertices align with those of $T$.

    Since we shift in a direction not aligned with any edge of $T$, and since $i \leq f-1$, each edge or vertex of $T$ lies on the boundary of $O(1)$ faces across the family $\{\T_i\}_{i=0}^{f-1}$. It follows that for any discrete vertex set $V \subseteq V(T)$ with $|V| = k$, there exists some $\T_i$ such that only $O(k/f)$ vertices lie on boundaries of faces in $\T_i$.
    Therefore, there exists a super-tiling $\T \in \{\T_0, \dots, \T_{f-1}\}$ such that the total number of boundary vertices from $P$ and $Q$ is $O((n + m)/f)$, which implies $O((n + m)/f) = O(\sqrt{n + m})$ induced subpaths.
    To compute $\T$, we iterate over all vertices $v$ of $P$ and $Q$ and, for each, determine the $O(1)$ super-tilings $\T_i$ with $v \in B(\T_i)$ by iterating over all $\T \in \{\T_0, \dots, \T_{f-1}\}$. This takes $O( (n + m)f )$ total time. We then select the super-tiling $\T_i$ that minimizes $
    |(P \cup Q) \cap B(\T_i)|$. 
\end{proof}

Let $\T$ be the computed super-tiling of $T$.
We partition $P$ and $Q$ into induced subpaths, denoted by $S(P, \T)$ and $S(Q, \T)$.
These $O(\sqrt{n+m})$ subpaths naturally induce a subdivision of the free-space diagram $M_\Delta$ into $O(n + m)$ rectangular submatrices $M_\Delta[i\colon i', j:j']$, where $P[i:i'] \in S(P, \T)$ and $Q[j:j'] \in S(Q, \T)$ (see Figure~\ref{fig:alg_overview}(a)).

Our algorithm makes use of a wave-front method. We define a wave-front $W$ as a monotone walk through the grid $[1, n] \times [1, m]$ that separates $(1,1)$ from $(n,m)$. At each grid point $(i, j) \in W$, we store a Boolean value indicating whether there exists a monotone walk $F$ from $(1,1)$ to $(i,j)$ such that all points on $F$ lie in free space: that is, $M_\Delta[x, y] = 0$ for all $(x, y) \in F$ (Figure~\ref{fig:alg_overview}(b)).

An \emph{update} to the wave-front $W$ selects a pair of subpaths $(P[i:i'], Q[j:j']) \in S(P, \T) \times S(Q, \T)$ such that the bottom and left facets of the submatrix $M_\Delta[i:i', j:j']$ lie on the current wave-front. The update removes these facets from $W$ and inserts the top and right facets instead.
To perform the update, we consider the faces $F_P$ and $F_Q$ of $\T$ that $P[i:i']$, respectively $Q[j:j']$, are assigned to, and proceed via a case distinction:

\begin{enumerate}
    \item If $F_P$ and $F_Q$ are not aligned, then they are well-separated by a vertex $v$ (by Lemma~\ref{lemm:aligned}), and we show how to perform a fast update using $v$.
    \item If $F_P$ and $F_Q$ are aligned, we perform a brute-force update and use a counting argument to bound the total running time of these updates.
\end{enumerate}

We analyse both cases separately, observing that while some instances may involve only Case 1 or only Case 2 submatrices, we can independently bound the total cost for each case.

\subparagraph{Case 1: Updating with $(P[i:i'], Q[j:j'])$ where $F_P$ and $F_Q$ are not aligned.}
We show an update algorithm that is near-linear in the perimeter of the submatrix $M_\Delta[i:i', j:j']$.

\begin{lemma}
\label{lem:update-not-aligned}
Let $P[i:i'] \in S(P, \T)$ and $Q[j:j'] \in S(Q, \T)$ be subpaths whose faces $F_P$ and $F_Q$ in $\T$ are not aligned. Then the corresponding wave-front update can be performed in $O(k \log k)$ time, where $k = |P[i:i']| + |Q[j:j']|$.
\end{lemma}

\begin{proof}
By Lemma~\ref{lemm:aligned}, the subcurves $P[i:i']$ and $Q[j:j']$ are well-separated by some vertex $v \in V(T)$, which we can identify in constant time.
Given $v$, we compute two curves in $\R$ whose $\Delta$-free space matrix equals $M_\Delta[i :i', j :j']$.
We define these curves as follows.

For a point $p \in P[i:i']$ we let $\bar{p} = \dist(p, v)$.
Likewise, for a point $q \in Q[j:j']$, we let $\bar{q} = -\dist(q, v)$.
Note the difference in sign.
Because $P[i:i']$ and $Q[j:j']$ are well-separated by $v$, we have $\dist(p, q) = \dist(p, v) + \dist(v, q) = |\bar{p} - \bar{q}|$ for any $p \in P[i:i']$ and $q \in Q[j:j']$.
We let $\bar{P} = (\bar{p}_i, \dots, \bar{p}_{i'})$ and $\bar{Q} = (\bar{q}_j, \dots, \bar{q}_{j'})$.

The $\Delta$-free-space matrix $M_\Delta[i:i', j:j']$ is equal to the $\Delta$-free space matrix of the sequences $\bar{P}$ and $\bar{Q}$ over the real line under the absolute value metric.
Let $\bar{M}_\Delta$ be this matrix. We compute the top and right facets of $\bar{M}_\Delta$ in $O(k \log k)$ time with existing fast algorithms for 1D separated sequences~\cite{bringmann17cpacked,vanderhorst25geodesic_frechet}.
\end{proof}

\subparagraph{Case 2: Updating with $(P[i:i'], Q[j:j'])$ where $F_P$ and $F_Q$ are aligned.}

For aligned face pairs, we fall back on a brute-force update strategy based on \emph{switching cells}, a concept introduced by Har-Peled, Knauer, Wang and Wenk.~\cite{AronovHKWW06}.

\begin{definition}[Switching cells]
A cell $(i,j)$ in $M_\Delta$ is a \emph{switching cell} if $M_\Delta[i,j] = 0$, but either $M_\Delta[i,j-1] = 1$ or $M_\Delta[i,j+1] = 1$ (or both).
\end{definition}

\noindent
The switching cells in a submatrix $M_\Delta[i:i', j:j']$ allow for a direct update of the wave-front:

\begin{lemma}[{\cite{AronovHKWW06}}]
    \label{lem:brute-force}
Let $P[i:i']$ and $Q[j :j']$ be two subpaths. Given all $s$ switching cells in $M_\Delta[i :i', j :j']$, a brute-force wave-front update takes $O(|P[i:i']| + |Q[j:j']| + s)$ time.
\end{lemma}

We apply the brute-force update algorithm to all subpath pairs $(P[i:i'], Q[j:j']) \in S(P, \T) \times S(Q, \T)$ whose associated faces $F_P$ and $F_Q$ in $\T$ are aligned. We now show that the total number of switching cells encountered in these submatrices is not too large:

\begin{lemma}
    \label{lem:intersection_count}
    Let $T$ be a regular unit tiling, and let $P$ and $Q$ be two paths in $T$. Let $\T$ be any super-tiling of $T$ with edge length $\Theta(\sqrt{n + m})$.
    Then there are at most $O(n \sqrt{n + m})$ pairs $(p_i, q_j)$ such that $M_\Delta[i, j]$ is a switching cell and $p_i$ and $q_j$ lie on aligned faces of $\T$.
    We can compute these pairs in $O(n \sqrt{n + m} \log m)$ time. 
\end{lemma}

\begin{proof}
  Fix a vertex $p_i$ and let $F_p$ be a face of $\T$ such that $p_i \in \overline{F_p}$. Since $P$ and $Q$ are paths in a unit tiling, a switching cell $M_\Delta[i, j]$ occurs only if $q_j$ lies at distance exactly $\Delta$ from $p_i$.

   Let $B_i(\Delta) \subseteq V(T)$ denote the set of all vertices at distance exactly $\Delta$ from $p_i$. This set lies on the boundary of a shape whose geometry depends on the tiling type:
    \begin{itemize}
        \item $B_i(\Delta)$ is the boundary of a diamond if $T$ is a square tiling.
        \item $B_i(\Delta)$ is the boundary of a regular hexagon if $T$ is a triangular tiling.
        \item $B_i(\Delta)$ is the boundary of a (non-regular) hexagon if $T$ is a hexagonal tiling.
    \end{itemize}
    
    These regions are convex. 
    Consider a face $F'$ of $\T$ that is aligned with $F_p$.
    The intersection of $B_i(\Delta)$ with $\overline{F'}$ contains at most $O(\sqrt{n + m})$ vertices. 
    Each of these vertices could correspond to at most one vertex $q_j \in Q$ where $(i, j)$ is a switching cell.
    There are only $O(1)$ faces of $\T$ that are aligned with $F_p$, whose closures intersect $B_i(\Delta)$.
    Hence, there are at most $O(\sqrt{n + m})$ such $q_j$ values per $p_i$. Summing over all $n$ vertices of $P$ gives the desired $O(n \sqrt{n + m})$ bound.

 What remains is to compute these pairs $(p_i, q_j)$.
We preprocess $Q$ in a membership query data structure.  
 For each $p_i \in P$, we select its assigned face $F_p$ in constant time and we identify the $O(1)$ faces $F'$ of $\T$ that are aligned with $F_p$ and intersect $B_i(\Delta)$ in constant time. 
 For every such face $F'$, we iterate over the vertices $v$ in  $\overline{F'} \cap B_i(\Delta)$. We find if there exists a $q_j \in Q$ with $v = q_j$ in $O(\log m)$ time through a membership query. 
 If so, then we compute $\dist(p_i, q_{j+1})$ and $\dist(p_i, q_{j-1})$ to determine whether $M_\Delta[i, j]$ is a switching cell.  
\end{proof}

\noindent
We are now ready to prove our main theorem:

\upperbound*

\begin{proof}
 Given $T$, $P$, and $Q$, we begin by applying Lemma~\ref{lem:properties} to compute in $O((n+m)^{1.5})$ time a super-tiling $\T$ and the corresponding subpaths $S(P, \T)$ and $S(Q, \T)$. The number of subpath pairs $(P[i:i'], Q[j:j']) \in S(P, \T) \times S(Q, \T)$ is $O(n + m)$.
    For each pair, we determine if the corresponding faces $F_P$ and $F_Q$ in $\T$ are aligned, which takes $O(n + m)$ total time. 
    
The product set $S(P, \T) \times S(Q, \T)$ induces a grid of $O(n + m)$ rectangular submatrices of $M_\Delta$ with $O(\sqrt{n + m})$ rows and columns. Thus:

    \begin{equation}
    \label{eq:upper_bound}
    \sum_{\textnormal{submatrices } M' \textnormal{ of } M_\Delta} \hspace{-1cm} |\textnormal{facets}(M')| = \sum_{(P[i: i'], Q[j: j']) \in S(P, \T) \times S(Q, \T)} \hspace{-0.8cm} (i' - i + j' - j)  \hspace{1cm} \in O( (n+m)^{1.5}).
     \end{equation}

    Next, we initialize the wave-front $W$ by computing the first row and first column of $M_\Delta$, i.e., all entries $M_\Delta[i, 1]$ for $i \in [n]$ and $M_\Delta[1, j]$ for $j \in [m]$, in $O(n + m)$ time.
We then iterate diagonally through all pairs $(P[i:i'], Q[j:j'])$ in $S(P, \T) \times S(Q, \T)$ and perform a wave-front update for each pair.

If the two faces $F_P$ and $F_Q$ of $\T$ assigned to $P[i: i']$ and $Q[j: j']$ are \emph{not aligned} then we apply Lemma~\ref{lem:update-not-aligned}. 
The total time spent in this case equals:

    \[
    \sum_{ (P[i: i'], Q[j: j']) \textnormal{ where } (F_P, F_Q) \textnormal{ are not aligned  } } \hspace{-2.5cm} O((i' - i + j' - j)) \log (i'-i+j'-j)  
\]
By Equation~\ref{eq:upper_bound}, this sum is upper bound by $O((n+m)^{1.5} \log (n + m)$.

    If the two faces $F_P$ and $F_Q$ of $\T$ assigned to $P[i: i']$ and $Q[j: j']$ are aligned then we apply Lemma~\ref{lem:brute-force} instead.
    The total time spent is:

   \[
    \sum_{P[i: i'], Q[j: j']) \textnormal{ where } (F_P, F_Q) \textnormal{ are aligned}} \hspace{-2cm}  O((i'-i+j'-j)  + \textnormal{switching cells}(M_\Delta[i:i',j:j']) )
\]
The first term and second term are bounded by $O((n + m)^{1.5} \log(n + m))$ by Equation~\ref{eq:upper_bound} and Lemmas~\ref{lem:brute-force}+\ref{lem:intersection_count}, respectively. 
Once the wave-front $W$ reaches the entry $(n, m)$, we have determined whether $\f(P, Q) \leq \Delta$, proving the theorem.
\end{proof}
    
\section{The discrete Fréchet distance under the $L_1$ metric}
\label{sec:L1}

We extend our techniques for paths in tilings to curves in the plane.
Fix some universal constant $\gamma > 0$ and fix some value $\eps \in (0,1)$. 
We consider $\mathbb{R}^2$ 
under the $L_1$ metric, and two polygonal curves $P$ and $Q$ with  $n$ and $m$ vertices, respectively.
In full generality, we only require that all edges in $P$ have length at most $\gamma$ and that $Q$ is an $(\eps, \delta)$-curve. 
Denote their discrete Fréchet distance under the $L_1$ metric by $D_F(P, Q)$. We show that we can compute $D_F(P, Q)$ in $\tilde{O}\left(\frac{\sqrt{\delta}}{\eps} (n+m)^{1.5} \right)$ time.
We adapt this algorithm in Appendix~\ref{app:euclidean} to make it work for any $L_c$ metric, by allowing some inaccuracy.

\subparagraph{Our approach.} Given some $\Delta \in \mathbb{R}$, we first show how to decide whether $D_F(P, Q) \leq \Delta$ using the wave-front propagation technique from Section~\ref{sec:upper_bound}, with suitable modifications for curves in the plane.
Let $t > 1$ be a parameter to be determined later. We define an axis-aligned square tiling $\T$ of the plane where each square has edge length $t$.

\begin{definition} We define the \emph{breakpoints} $B(P, \T)$ of $P$ as the set of vertices $p \in V(P)$ such that some edge of $P$ incident to $p$ intersects an edge of $\T$. \end{definition}

Note that $B(P, \T)$ may contain overlapping but distinct vertices of $P$ (since an $(\eps,\delta)$-curve may visit the same point up to $\delta$ times). 
If we remove all edges from $P$ which intersect some edge of $T$, we see that the breakpoints induce a partition of the vertices of $P$ into pairwise vertex-disjoint subcurves $P[i:i']$, 
each fully contained in the closure $\overline{F}$ of a face $F \in \T$. 
We refer to these subcurves as the \emph{induced subcurves} $S(P, \T)$ and we assign each induced subcurve to a face in $\T$. 
We prove that for a suitable choice of $\T$, the total number of induced subcurves satisfies $|S(P, \T) \cup S(Q, \T)| = O((n+m)/t)$.

We then apply our wave-front algorithm over the product set $S(P, \T) \times S(Q, \T)$. For a pair of subcurves $(P[i:i'], Q[j:j'])$, let $F_P$ and $F_Q$ denote their respective containing faces in $\T$. We distinguish two cases:

\begin{itemize}
    \item If $F_P$ and $F_Q$ are not aligned, we show that $(P[i:i'], Q[j:j'])$ is well-separated by a corner of either $F_P$ or $F_Q$. We then apply Lemma~\ref{lem:update-not-aligned} to perform a wave-front update in $O((i'-i + j'-j) \log(n+m))$ time.
    \item If $F_P$ and $F_Q$ are aligned, we apply the brute-force algorithm of Lemma~\ref{lem:brute-force}, but with a more careful analysis of the number of switching cells.
\end{itemize}

\noindent
Balancing the choice of $t$ is critical: a larger $t$ reduces the number of wave-front updates, while a smaller $t$ reduces the number of switching cells encountered. We choose $t$ to optimize the overall running time.

\subparagraph{Choosing a tiling.}
We construct a tiling $\T$ such that it creates few induced subcurves.

\begin{lemma}
\label{lem:properties_L1}
Given $t \geq 1$, we can compute in $O((n+m)t)$ time a square tiling $\T$ of edge length $t$ such that the number of induced subcurves satisfies $|S(P, \T) \cup S(Q, \T)| = O((n+m)/t)$.
\end{lemma}

\begin{proof}
The proof mirrors that of Lemma~\ref{lem:properties}. Let $\T_0$ be an axis-aligned square tiling with edge length $t$, and let $\T_i$ denote the tiling obtained by shifting $\T_0$ by $(i, i)$ for $i \in [t-1]$. We required that each edge $e$ of $P$ or $Q$ has length at most $\gamma \in O(1)$. Thus each edge intersects the boundary of at most $O(1)$ of the shifted tilings. It follows that there exists an index $i$ such that $\T_i$ induces only $O((n+m)/t)$ breakpoints across $P$ and $Q$.

To find $\T_i$, we iterate over all $O(n+m)$ edges of $P$ and $Q$. For each edge, we test for intersection with each of the $t$ candidate tilings in constant time. This gives an overall running time of $O((n+m)t)$ to find the tiling $\T_i$ that minimizes $|B(P, \T_i) \cup B(Q, \T_i)|$.
\end{proof}

\subparagraph{Wave-front updates.} Unlike in Section~\ref{sec:upper_bound}, the subcurves in $S(P, \T)$ are vertex-disjoint. Thus, we never have a submatrix $M_\Delta[i:i', j:j']$ whose bottom-left facet coincides with the current wave-front $W$. However, we can always find a submatrix $M' := M_\Delta[i:i', j:j']$ where each of the cells on the bottom-left facet of $M'$ is adjacent to a cell in $W$. We extend $W$ to cover it in $O(i'-i + j'-j)$ time via brute force and thereby define wave-front updates analogously.

\subparagraph{Case 1: Updating with $(P[i:i'], Q[j:j'])$ where $F_P$ and $F_Q$ are not aligned.}
Our key observation is that, even though the edges of $P[i:i']$ and $Q[j:j']$ do not coincide with $\T$, a pair of faces $(F_P, F_Q)$ of $\T$ are still well-separated whenever they are not aligned:

\begin{lemma}
\label{lem:update-not-aligned-L1}
     Let $P[i:i'] \in S(P, \T)$ and $Q[j:j'] \in S(Q, \T)$ be subcurves whose faces $F_P$ and $F_Q$ in $\T$ are not aligned. Then the corresponding wave-front update can be performed in $O(k \log k)$ time, where $k = |P[i:i']| + |Q[j:j']|$.
\end{lemma}

\begin{proof} Assume, without loss of generality, that the bottom-left corner of $F_P$ lies above and to the right of the top-right corner of $F_Q$. Then any shortest $L_1$ path from a point in $\overline{F_P}$ to one in $\overline{F_Q}$ may be rerouted through that corner. The rest follows from Lemma~\ref{lem:update-not-aligned}. \end{proof}

\subparagraph{Case 2: Updating with $(P[i:i'], Q[j:j'])$ where $F_P$ and $F_Q$ are aligned.}
We upper bound the number of switching cells $M_\delta[i, j]$ where $p_i$ and $q_j$ lie in aligned faces of $\T$:

\begin{lemma}
    \label{lem:intersection_count_L_1}
    Let $\T$ be the tiling from Lemma~\ref{lem:properties_L1}. 
    Then there are at most $O(\frac{\delta n t}{\eps^2})$ pairs $(p_i, q_j)$ such that $M_\Delta[i, j]$ is a switching cell and $p_i$ and $q_j$ lie in aligned faces of $\T$.
    We can compute these pairs in $O(\frac{\delta n t}{\eps^2} \log m)$ time. 
\end{lemma}
\begin{proof}
    We preprocess $Q$ by snapping its vertices to a square grid $G$ of edge length $\eps^{-1}$, forming $Q'$.
    Each vertex in $Q'$ corresponds to $O(\delta)$ original vertices since $Q$ is an $(\eps, \delta)$-curve.
    We store $Q'$ in a spatial data structure (e.g., a lexicographically sorted balanced binary tree), where a query point $z$ returns all $q_j \in Q$ corresponding to $z$ in $O(\delta + \log n)$ time.

    Recall that all edges of $Q$ have length at most $\gamma \in O(1)$. 
    Fix a vertex $p_i$ and let $F_p$ be a face of $\T$ that contains $p_i$.       
    Let $B_i(\Delta)$ be the metric circle (under the $L_1$ metric) of radius $\Delta$ centred at $p_i$. Let $\Gamma$ denote a ball with radius $2 \gamma$ and $B_i^*(\Delta)$  be the metric annulus that is the Minkowski-sum $\Gamma \oplus B_i(\Delta)$. 
    Observe that an edge $(q_j, q_{j+1})$ of $Q$ intersects $B_i(\Delta)$ only if one of its corresponding vertices $q_j'$ and $q_{j+1}'$ of $Q'$ lies in $B_i^*(\Delta)$. 
    
    There are $O(1)$ faces $F'$ of $\T$ that are aligned with $F_p$ and that intersect $B_i^*(\Delta)$. 
    The intersection of $B_i^*(\Delta)$ with $F'$ contains at most $O(\frac{\gamma t}{\eps^2}) = O(\frac{t}{\eps^2})$ vertices of $G$.
    We iterate over the vertices in $B_i^*(\Delta) \cap \overline{F'} \cap V(G)$ and for each we perform a membership query on $Q'$. 
    This returns $O(\frac{\delta t}{\eps^2})$ vertices of $Q'$.
    For each reported vertex $q'_j$ we test if $M_\Delta[i, j]$ is a switching cell in constant time.
    Thus, for any $p_i \in P$, we compute all switching cells $M_\Delta[i, j]$  that correspond to the lemma statement in $O(\frac{\delta t}{\eps^2} \log m)$ time. 
\end{proof}

\begin{restatable}{theorem}{upperbound_L1}
    \label{thm:upperbound_L1}
    Let $P$ and $Q$ be two curves in $\mathbb{R}^2$ under the $L_1$ metric with $|P| = n$ and $|Q| = m$.
    Fix a universal constant $\gamma \in O(1)$.
    For any $\eps > 0$ and $\Delta \in \mathbb{R}$, if $P$ has edge length at most $\gamma$ and $Q$ is an $(\eps, \delta)$-curve, then we can compute whether $D_F(P, Q) \leq \Delta$ in $O(\frac{\sqrt{\delta}}{\eps} (n+m)^{1.5} \log (n+m))$ time.  
\end{restatable}

\begin{proof}
Let $t \geq 1$ be a value that we set later.
Given $P$ and $Q$, we apply Lemma~\ref{lem:properties_L1} to compute in $O( (n + m)t)$ time our tiling $\T$ and the $O( (n+m)/t)$ induced  subcurves $S(P, \T)$ and $S(Q, \T)$.    
The product set $S(P, \T) \times S(Q, \T)$ induces a grid over $M_\Delta$ with  $O( (n+m)/t)$ rows and columns. Thus: \vspace{-0.2cm}

    \begin{equation}
    \label{eq:upper_bound-L1}
    \sum_{\textnormal{submatrices } M' \textnormal{ of } M_\Delta} \hspace{-1cm} |\textnormal{facets}(M')| = \sum_{(P[i: i'], Q[j: j']) \in S(P, \T) \times S(Q, \T)} \hspace{-0.8cm} (i' - i + j' - j)  \hspace{1cm} \in O( (n+m)^2/ t).
     \end{equation}

We iteratively perform wave-front updates with the input $(P[i: i'], Q[j: j'])$. 
Let $F_P$ and $F_Q$ denote the corresponding two faces of $\T$. 
If $F_P$ and $F_Q$ are \emph{not aligned} then we apply Lemma~\ref{lem:brute-force}. 
The total time spent in this case equals:

    \[
    \sum_{ (P[i: i'], Q[j: j']) \textnormal{ where } (F_P, F_Q) \textnormal{ are not aligned  } } \hspace{-2.5cm} O((i' - i + j' - j)) \log (n+m)  
\]
By Equation~\ref{eq:upper_bound-L1}, this sum is upper bound by $O(\frac{(n+m)^{2}}{t} \log (n + m))$.

    If $F_P$ and $F_Q$ are aligned then we apply Lemma~\ref{lem:intersection_count} instead.
    The total time spent is:

   \[
    \sum_{P[i: i'], Q[j: j']) \textnormal{ where } (F_P, F_Q) \textnormal{ are aligned}} \hspace{-2cm}  O((i'-i+j'-j)  + \textnormal{switching cells}(M_\Delta[i:i',j:j']) )
\]
The first term is upper bounded by  $O(\frac{(n+m)^{2}}{t})$ by Equation~\ref{eq:upper_bound-L1}.
The second term is upper bounded by $O( \frac{\delta n t}{\eps^2} \log m)$ by combining Lemma~\ref{lem:brute-force} and \ref{lem:intersection_count_L_1}.
It follows that if we choose $t \in \Theta(\frac{\eps \cdot (n+m)}{\sqrt{\delta n}})$ the total running time is upper bounded by $O(\frac{\sqrt{\delta}}{\eps}(n+m)^{1.5}  \log (n+m))$.
\end{proof}

\subparagraph{Computing $D_F(P, Q)$.} Given $P$ and $Q$, there is a set of $O(nm)$ values $S$, which can be computed in $O(nm)$ time, such that $D_F(P, Q) \in S$.
The classical algorithm to compute $D_F(P, Q)$ performs binary search over $S$. 
Since we want to use subquadratic time, we cannot compute the set $S$ explicitly.
Instead, we represent $S$ as the \emph{Cartesian sum} $X \oplus Y$ of two sets of $O(n+m)$ values.
The Cartesian sum of $X$ and $Y$ is the multiset $X \oplus Y=\{\!\{x+y\mid x\in X, y\in Y\}\!\}$.
We can use existing techniques~\cite{Frederickson84selection_XY,mirzaian85selection_XY} for selection in Cartesian sums to binary search over the Cartesian sum $X \oplus Y$ without explicitly computing its $O(|X|\cdot|Y|)$ elements.

Specifically, we define $X$ and $Y$ as follows.
The set $X$ contains, for each point $(x, y) \in P$, the four values $x+y$, $-x+y$, $x-y$ and $-x-y$.
The set $Y$ is defined analogously for points in $Q$.
For all $(p, q) \in P \times Q$, the distance between $p$ and $q$ is the sum of an element in $X$ and an element in $Y$.
That is, their distance is an element of $X \oplus Y$.
Hence, $D_F(P, Q) \in X \oplus Y$.

After sorting $X$ and $Y$, we can compute the $i^{\mathrm{th}}$ smallest value in $X \oplus Y$, for any given $i$, in $O(|X|+|Y|) = O(n+m)$ time~\cite{Frederickson84selection_XY,mirzaian85selection_XY}.
We binary search over the integers $1, \dots, |X| \cdot |Y|$, and for each considered integer $i$, we compute the $i^{\mathrm{th}}$ smallest value $\Delta$ in $X \oplus Y$.
Then we use our decision algorithm (Theorem~\ref{thm:upperbound_L1}) to decide whether $D_F(P, Q) \leq \Delta$ to guide the search to $D_F(P, Q)$. Thus, we obtain the following result:

\begin{corollary}
Given $\gamma > 0$,
let $P$ and $Q$ be curves  in $\mathbb{R}^2$ under $L_1$ with $|P|+|Q|=n+m$.
%$|P| = n$ and $|Q| = m$.
    For any $\eps > 0$, if $P$ has edge length at most $\gamma$ and $Q$ is an $(\eps, \delta)$-curve, then $D_F(P, Q)$ can be computed in:
    $O(\frac{\sqrt{\delta}}{\eps} (n+m)^{1.5} \log^2 (n+m))$ time.  
\end{corollary}

\bibliographystyle{plainurl}
\bibliography{references}

\appendix

\section{A lower bound in unweighted planar graphs}
\label{app:lower_bound}

In this section we describe the lower bound construction based on the orthogonal vector hypothesis (OVH). The goal is to construct from some given OV instance in linear time a graph $G$ and two paths $P,Q$ such that $\frechet(P,Q) \leq 4$ if and only if the OV instance is a yes-instance and $\frechet(P,Q) \geq 5$ otherwise. This suffices to prove Theorem~\ref{thm:main-lower-bound}.
We start by preprocessing the instance, which significantly helps us at reducing the technical complexity of the proof.
For this, if $V = \set{v^{(1)}, \dots, v^{(n)}} \subseteq \{0,1\}^d$ is a set of binary vectors, consider the binary string $\str(V) := v^{(1)}v^{(2)}\dots v^{(n)} \in \set{0,1}^{nd}$ created by writing down all vectors in $V$ one after another.
A \emph{substring} of a string $s = s_1 s_2\dots s_k$ is a sequence of consecutive characters of $s$, that is a string of the form $s_as_{a+1}\dots s_{b-1} s_b$ for some $1 \leq a \leq b \leq k$.

\preprocessing*

\begin{proof}
    For better readability, in this proof we use the notation of binary strings and binary vectors interchangeably, i.e.\ $001$ refers to the vector $(0,0,1)$. 
    Let $U' = \set{u'^{(1)},\dots, u'^{(n)}} \subseteq \set{0,1}^{d'}$ and $W' = \set{w'^{(1)},\dots, w'^{(m)}} \subseteq \set{0,1}^{d'}$ be the vector sets of the original OV instance.

    First, we create sets of vectors $\tilde U, \tilde W \subseteq \set{0,1}^{2d'+1}$ from $U',W'$.
    Then, we create sets of vectors $U,W \subseteq \set{0,1}^{2d'+4}$ from $\tilde U, \tilde W$ with $|U| = |\tilde U| = |U'| = n$ and $|V| = |\tilde V| = |V'| = m$ such that for $d := 2d'+4$, the $d$-dimensional OV instance $(U, W)$ has the desired properties.

    For the first step, for every $u' \in U'$, we add the vector $\tilde u$ to $\tilde U$ obtained from $u'$ by interweaving $u'$ with the all-zero vector, such that the first and last coordinate of the new $(2d+1)$-dimensional vector $\tilde u$ is a $0$. For example, if $u' = 1101$, then $\tilde u = 010100010$.
    For all $w' \in W'$, we do an analogous process to obtain $\tilde W$, with the exception that we interweave it with the all-ones vector instead.
    For example if $w' = 0100$, then $\tilde w = 101110101$.
   % Since $u' \in U'$ is interweaved with ones,      and $w' \in W'$ is interweaved with zeros, we have that $u'$ and $w'$ are orthogonal if and only if $u$ and $w$ are orthogonal. Note that every $w' \in W'$ does not contain two consecutive zeros.

    In the second step, for every $\tilde u \in \tilde U$, we obtain a vector $u \in \set{0,1}^{2d+4}$, by adding the \enquote{header} sequence $011$ at the beginning of $\tilde u$.
    For example, if $\tilde u = 010100010$, then $u = 011010100010$.
    For $\tilde w \in \tilde W$, we define $w$ analogously by adding the header $100$ instead. 
    Observe that since the two different headers are orthogonal and the interweaved bits are $0$ and $1$ respectively, 
    $u \in U$ is orthogonal to $w \in W$ if and only if $u'$ and $w'$ are orthogonal.     
    Furthermore, every $u \in U$ starts and ends with a $0$, while every $w \in W$ starts and ends with a $1$.
    This shows that the instance $(U,W)$ has the first two claimed properties.
    Finally, we prove that $(U, W)$ also has the third property. Consider the string $\str(W)$, it contains exactly $m$ substrings $00$. Indeed, every header of some $w \in W$ contains one occurrence of $00$. 
    But these are the only occurrences of $00$, since every $w \in W$ starts and ends with a $1$ and interweaves a $1$-bit between every bit from the original vector $w'$.     
    Now, assume that some $u \in U$ is orthogonal to a substring $x$ of length $2d'+4$ of $\str(W)$. Since $u$ has two consecutive ones in the header, they must be perfectly aligned with two consecutive zeros in $\str(W)$. 
    But this already implies that the header of $u$ must be perfectly aligned with the header of some $w \in W$, and hence $u$ and $w$ are orthogonal. 
\end{proof}

\subsection{Vector gadgets}

From now on, we assume that the input instance $U,W \subseteq \{0,1\}^d$ of OV meets the guarantees of the preprocessing. We are now ready to describe the lower bound construction. Given $(U,W)$ as input, the reduction constructs a graph $G$ and paths $P,Q$.
We start by describing certain subgraphs of $G$ with helpful properties, which we call \emph{vector gadgets}.
Specifically, we introduce two kinds of gadgets:
Given two vectors $u \in U, w \in W$, we construct from $u$ a so-called \emph{orange vector gadget} for the person, denoted by $G_u$. Likewise, we construct from $w$ a so-called \emph{blue vector gadget} for the dog, denoted by $G_w$.
We prove that the gadgets $G_u$ and $G_w$ have the following crucial property: 
Under the assumption that the person and dog are already positioned at the start of $G_u$ and $G_w$, the person and dog can traverse the gadgets maintaining a distance of 4 if and only if $u$ and $v$ are orthogonal.

\begin{figure}
    \centering
    \includegraphics[page=1,scale=0.9]{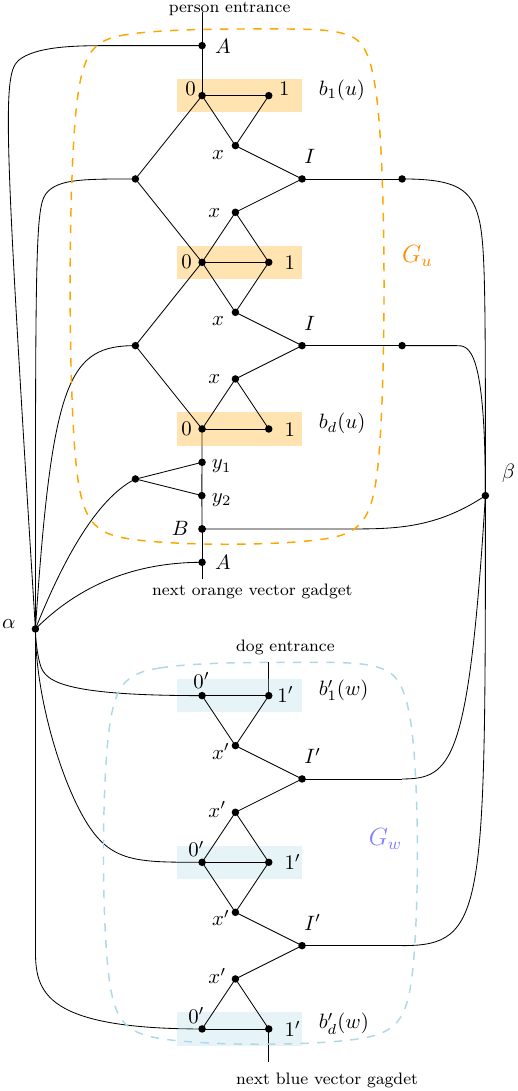}

    \vspace*{0.5cm}
    \begin{tabular}{c|cccccccc}
         & $0$ & $1$ & $x$ & $I$\\
         \hline
         $0'$ & 3 & 4 & 4 & \red 5 \\
         $1'$ & 4 & \red 5 & \red 5 & \red 5 \\
         $x'$ & 4 & \red 5 & \red 5 & 4 \\
         $I'$ & \red 5 & \red 5 & 4 & 3
    \end{tabular}
    \caption{Vector gagdets for the person (orange) and the dog (blue) of length three, and a corresponding distance table.}
    \label{fig:vector-gadget}
\end{figure}

\begin{figure}
    \centering
    \includegraphics[page=2,scale=0.9]{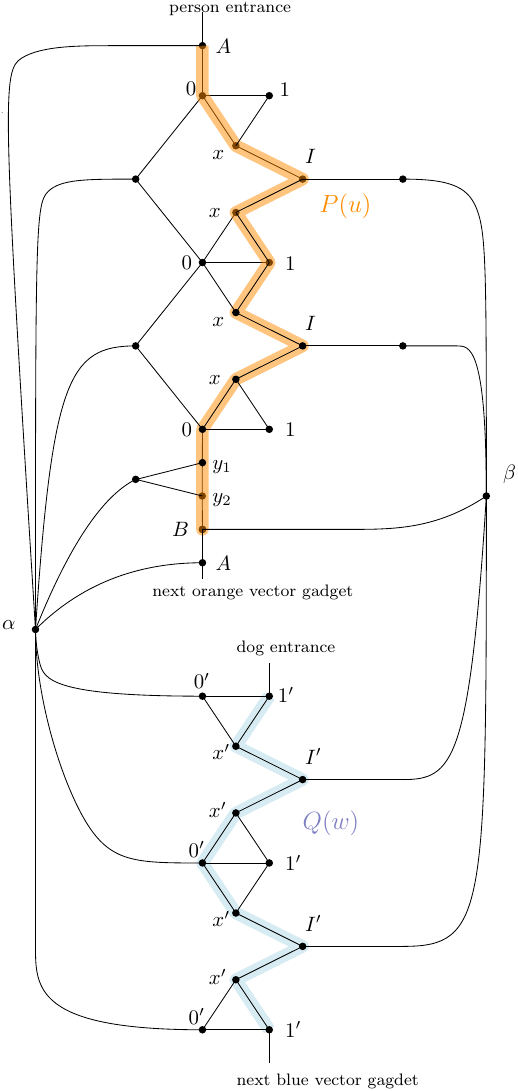}
    \caption{Example of encoding the vector $u = (0,1,0)$ in the orange gadget and the vector $w = (1,0,1)$ in the blue gadget. Note that the paths traverse exactly one vertex of each box, thereby encoding a binary vector.}
    \label{fig:vector-encoding-as-paths}
\end{figure}

Formally, let us define $G_u$ and $G_w$ as the subgraphs depicted in Figure~\ref{fig:vector-gadget}. 
The construction comes with equivalence classes of vertices labelled either $0, 1, x, I, A, B, y_1, y_2$ in the orange gadgets, or $0', 1', x', I'$ in the blue gadgets, and two special vertices called $\alpha$ and $\beta$. 
The pairwise distance between vertices of the blue gadgets and vertices of the orange gadgets is given in the table in Figure~\ref{fig:vector-gadget}. 
We will later make sure that when we use these gadgets as part of the larger graph $G$, that none of the distances described in the distance table changes.
Let us call a pair of vertices $(0,1)$ or $(0',1')$ as highlighted in Figure~\ref{fig:vector-gadget} a \emph{box}. The length of a vector gadget is the number of boxes in it. For vectors $u,w \in \set{0,1}^d$, the gadgets $G_u, G_w$ always have length $d$. 
We denote by $b_1(u),\dots, b_d(u)$ the $d$ boxes in gadget $G_u$. 
Likewise, we denote by $b'_1(w),\dots, b_d(w)$ the $d$ boxes in gadget $G_w$. 

Furthermore, we define the path $P(u)$ (see Figure~\ref{fig:vector-encoding-as-paths}) as the path inside of gadget $G_u$ which starts at vertex $A$ and ends with the vertex sequence $y_1,y_2,B$. In between, it traverses all boxes of $G_u$, and for all $i \in [d]$ inside of every box $b_i(u)$ visits vertex 1 if and only if $u_i = 1$. 
In between the boxes the path $P(u)$ visits the  intermediary vertices $x,I,x$.

The restriction of $P(u)$ to the orange boxes within the gadget $G_u$ encodes the vector $u$. 
In general, we will define a path $P$ through the graph $G$ that traverses each vector gadget once, such that for all orange vector gadgets $G_u$, the path $P$ restricted to $G_u$ equals $P(u)$.  
Likewise we can encode a vector $w \in \set{0,1}^d$ inside a blue vector gadget $G_w$ of length $d$ as a subpath $Q(w)$. 
The path $Q(w)$ starts at the first box $b'_1(w)$ and ends at the last box $b'_d(w)$.
Note that $u_1 = u_d = 0$ and $w_1 = w_d = 1$ due to the preprocessing. Hence the path $P(u)$ visits the 0 in the boxes $b_1(u), b_d(u)$, and the path $Q(w)$ visits the 1 in the boxes $b'_1(w), b'_d(w)$.

For the next lemma we require one more concept. Let $G_u$ be an orange vector gadget.
For some index $i \in \set{1,\dots, d-1}$, the term $\nnext(b_i(u)) := b_{i+1}(u)$ denotes the next box after $b_i(u)$. For the last box in the gadget, $\nnext(b_d)(u)$ is undefined.
Likewise, if $G_w$ is some blue vector gadget of length $d$, we let $\nnext(b'_i(w)) := b'_{i+1}(w)$ for all $i \in \set{1,\dots, d-1}$.
The reader may assume that the term $\nnext(b'_d(w))$ is undefined for now. 
However, we will later show that we can construct the graph $G$ (see \cref{fig:full-construction}) in such a way that for the blue vector gadgets it is natural to extend this definition to $\nnext(b'_d(w^{(k)})) := b'_1(w^{(k+1)})$ for all $k \in [m-1]$. Here, $W = \set{w^{(1)},\dots,w^{(m)}}$.

Let $\mathcal{B}$ be the set of all orange boxes and  $\mathcal{B'}$ be the set of all blue boxes, i.e.\
\[
\mathcal{B} := \bigcup_{u \in U} \set{b_1(u),\dots, b_d(u)}, \quad \mathcal{B}' := \bigcup_{v \in V} \set{b'_1(v),\dots, b'_d(v)}.
\]
Two boxes $b_i(u) \in \mathcal{B}$ and $b'_j(v) \in \mathcal{B'}$ are called \emph{locally orthogonal}, 
if $u_iv_j = 0$, where $u_i, v_j$ are the $i$-th and $j$-th component of the vectors $u,v$. This is equivalent to saying that at least one of the paths $P(u)$ and $Q(v)$ crosses a 0 or $0'$ in the box $b_i(u)$ or $b'_j(v)$.

\begin{lem}
\label{lem:gadget-sufficient-single-step}
Let $b \in \mathcal{B}, b' \in \mathcal{B'}$ be boxes such that both $\nnext(b), \nnext(b')$ are defined. Assume the person is sitting on the 0 or 1 in the box $b$, and the dog is sitting on the 0 or 1 in the box $b'$. If $b,b'$ are locally orthogonal and $\nnext(b), \nnext(b')$ are locally orthogonal, then there is a way for the person and the dog to traverse their paths, while maintaining a distance of 4, such that after finitely many steps they are simultaneously sitting in $\nnext(b)$ and $\nnext(b')$.   
\end{lem}

In particular, this lemma implies the following: 
If the person and the dog are positioned such that the person is currently on the first vertex of $P(u)$ for some vector $u$, and the dog is on the first vertex of $Q(w)$ for some vector $w$, and $(u, w)$ are orthogonal, 
then they can simultaneously traverse $P(u)$ and $Q(w)$, while maintaining a distance of 4.

\begin{proof}[Proof of Lemma~\ref{lem:gadget-sufficient-single-step}]
Since the boxes are locally orthogonal, either the person or the dog is currently sitting on a 0. 
We construct a traversal.
First, the one sitting on a 0 stays put, while the other one makes a step. 
This is legal, since $\dist(0, x') = \dist(x, 0') = 4$.
Now both make a step. This is legal, since $\dist(x, I') = \dist(I, x') = 4$.
The one lagging behind makes a step, so that now both are on $I$ and $I'$.
Since the vectors are locally orthogonal at the next box, at least one of the person or the dog has a 0 in the next box. 
The being with a 0 in the next box makes a step, while the other stays put. This is again legal because $\dist(x, I') = \dist(I, x') = 4$.
Both simultaneously make a step. Since the being with the 0 in the next box went first, they are now on a 0, and we have $\dist(0, x') = \dist(x, 0') = 4$. Finally, the being lagging behind makes the last step, so now both beings are inside the next box (which is legal since $\dist(0,0'), \dist(0, 1'), \dist(1, 0') \leq 4$). 
\end{proof}

The sufficient condition of Lemma~\ref{lem:gadget-sufficient-single-step} is accompanied by the following necessary condition. Given some orange box $b$, let us say the person is \emph{adjacent to $b$}, 
if he is positioned either inside $b$, or at a vertex $x$ immediately before or after $b$ (but not at $A$ or $y_1$, if $b$ happens to be the first/last box).
Given some blue box $b'$, let us say the dog is \emph{adjacent to $b'$}, if he is inside $b'$, or at a vertex $x'$ immediately before or after $b'$. 
The following lemma essentially states that if at some point of time the person and dog are both inside a vector gadget, 
they are forced to traverse both vector gadgets simultaneously at \enquote{the same speed}.

\begin{restatable}{lem}{gadgetsinglestep}
\label{lem:gadget-necessary-single-step}
    Let $b \in \mathcal{B}, b' \in \mathcal{B}'$ be boxes such that both $\nnext(b), \nnext(b')$ are defined.
    If currently the person is adjacent to $b$ and the dog is adjacent to $b'$, 
    and they traverse their paths while maintaining a distance of at most 4, 
    then after finitely many steps we have that simultaneously the person is adjacent to $\nnext(b)$ and the dog is adjacent to $\nnext(b')$. 
    Furthermore, if $\nnext(b), \nnext(b')$ are not locally orthogonal, then it is impossible for the person and dog to continue on their paths.
\end{restatable}

\begin{proof}
    Since $\dist(x,x') = 5$, it is not the case that both the person and the dog are currently on $x$ or $x'$. 
    Hence at least one of the man and the dog is currently inside a box.
    We \emph{claim} that we can assume w.l.o.g. that one of the two is on the vertex 0 in their box $b$ or $b'$ 
    while the other one is on the vertex $x'$ immediately after box $b$ or $b'$.
    Indeed, in the first case, where both of them are in the box, they cannot move simultaneously, since $\dist(x,x') = 5$. So only one of them can move, while the other stays in the box. 
    Since $\dist(1,x') = \dist(1', x) = 5$, the one staying in the box must be on a 0.
    In the second case (where one of them is inside the box, and the other one is lagging behind on the $x$ before a box) the one inside the box cannot move. 
    So, they either have to move simultaneously or only the one lagging behind moves on their own. 
    In both cases we arrive at a situation we analysed before. This proves the claim. 
    
    We can hence assume one of them is in box $b$ or $b'$ and the other on is on $x$ or $x'$ after the box $b$ or $b'$.
    Now, again, the one lagging behind cannot move on their own. The one standing on $x$ or $x'$ can also not move on their own, since $\dist(0, I') = \dist(0', I) = 5$ and $\dist(1, I') = \dist(1', I) = 5$.
    So both have to move simultaneously.
    Hence the person is on $x$ while the dog is on $I'$, or the other way around. 
    The being in front cannot move on their own, so the one lagging behind has to move, or both have to move simultaneously.
    For a reasoning similar to above, after at most two rounds, exactly one of the two is on $I$ or $I'$, while the other is at $x$ or $x'$ before the box $\nnext(b)$ or $\next(b')$.
    Now, the only way to make progress is if both move at the same time. Note however that after this step one is in the box $\nnext(b)$ or $\nnext(b')$ while the other one is on the vertex $x$ or $x'$ just before $\nnext(b)$ or $\nnext(b')$. 
    This proves the first part of the lemma. 
    For the second part, observe that in this final situation, the one inside the box $\next(b)$ or $\next(b')$ must be on a 0, since $\dist(1,x') = \dist(1',x) = 5$. Hence the boxes $\nnext(b), \nnext(b')$ are locally orthogonal by definition.
\end{proof}

\subsection{The full construction}
\begin{figure}
    \centering
    \includegraphics[scale=0.9]{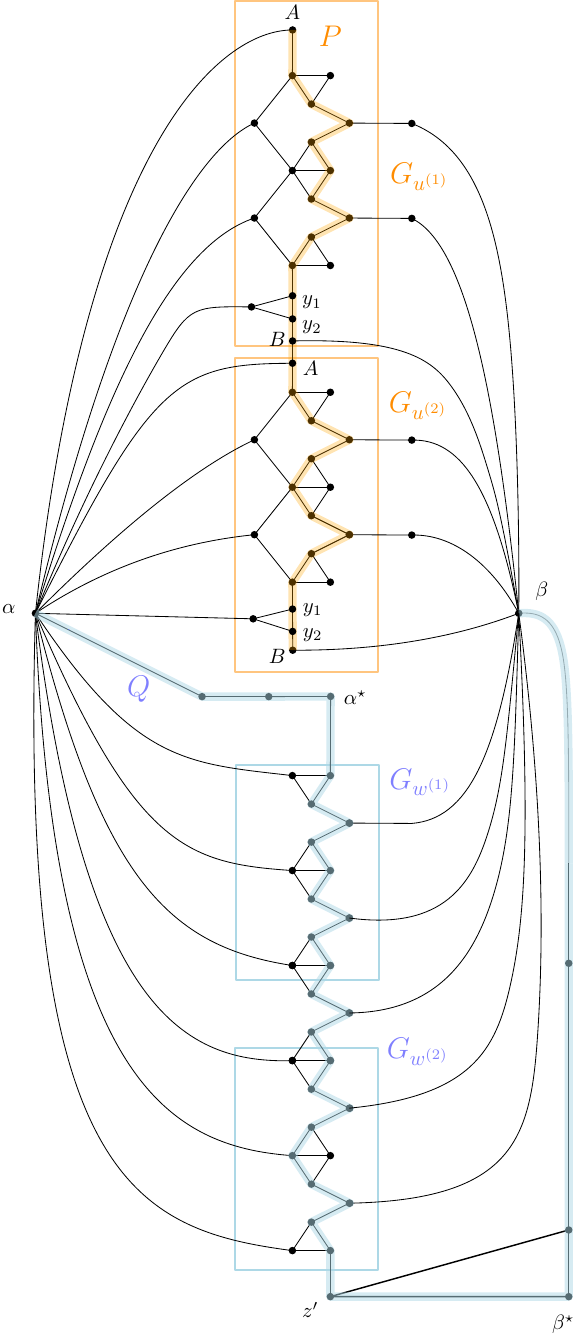}
    \caption{Complete construction of the reduction from OVH, for the two sets $U = \set{(0,1,0), (0,0,0)}$ and $W = \set{(1,1,1), (1,0,1)}$.}
    \label{fig:full-construction}
\end{figure}

We describe the complete reduction from the OV problem to the discrete Fréchet distance. 
We are given an instance of $U,W \subseteq \R^d$ with $U = \set{u^{(1)},\dots,u^{(n)}}$ and $W = \set{w^{(1)},\dots, w^{(m)}}$. We assume that it satisfies the guarantees of the preprocessing. 
We define a graph $G$ and paths $P, Q$ based on $U,W$ as in Figure~\ref{fig:full-construction}. 
For each vector $u \in U$ the graph contains a unique orange vector gadget $G_u$. 
For each vector $w \in W$ the graph contains a unique blue vector gadget $G_w$.
Consecutive orange vector gadgets are connected by an edge.
Consecutive blue gadgets are connected by three vertices $x',I',x'$ in the same manner as two boxes inside the gadgets.  
Finally, vertices $\alpha, \alpha^\star, z', \beta, \beta^\star$, and paths between $\alpha$ and $\alpha^\star$ as well as between $\beta$ and $\beta^\star$ are added to the graph. 
We connect $\alpha^\star$ to the 1 in the first box of $G_{w^{(1)}}$, and $z'$ to the 1 in the last box of $G_{w^{(d)}}$ and the vertex after $\beta^\star$.
This completes the description of the planar graph $G$.
The orange path $P$ starts at $A$ before the first orange gadget and ends at $B$ after the last orange gadget and traverses the orange gadgets in-between. 
The blue path $Q$ goes from $\alpha$ to $\alpha^\star$, traverses the blue gadgets, then goes from $z'$ to $\beta^\star$ to $\beta$.
This completes the description of $P$ and $Q$.

Observe the following properties: 
As claimed before, the distance table between vertices in orange and blue gadgets is the same in the larger graph $G$ as in \cref{fig:vector-gadget}. The distance depends only on the equivalence class, not the individual vertices. 
This can be formally proven by noting that every shortest path uses either $\alpha$ or $\beta$, and all vertices from the same equivalence class have the same distance to both $\alpha$ and $\beta$.
Vertices $A$ or $B$ have all vertices of all blue vector gadgets entirely within distance 4. 
Vertices $\alpha$ and $\beta$ have all vertices of all orange vector gadgets entirely within distance 5.
On the other hand, if the dog is standing on $\alpha^\star$, the only vertices of $P$ within a reach of 4 are vertices labelled $A$. 
Similarly, if the dog stands on $\beta^\star$, the only vertices of $P$ within a reach of 4 are vertices labelled $B$.

\begin{lem}
\label{lem:lower-bound-sufficient}
    If $(U,W)$ is a yes-instance of OV, then $\frechet(P, Q) \leq 4$.
\end{lem}

\begin{proof}
If $u \in U$ is orthogonal to $w \in W$, the person and the dog can employ the following 6-phase strategy:
\begin{enumerate}
    \item The dog stays at $\alpha$, while the person goes to the orange gadget $G_u$ corresponding to $u$ and waits at $A$ before gadget $G_u$.
    \item Then the dog walks over $\alpha^\star$ to the blue gadget $G_w$ corresponding to $w$.
    \item Then, the person and the dog simultaneously move to the first box of their respective gadget, and traverse their gadget in unison (which is possible due to Lemma~\ref{lem:gadget-sufficient-single-step} since $u,w$ are orthogonal).
    \item The dog waits at vertex 1 in box $b'_d(w)$ while the person goes via $y_1$ and $y_2$ to the vertex $B$ immediately after $G_u$.
    \item  Then the person waits at $B$ while the dog walks all of the remaining path $Q$ to $\beta$, passing over $z'$ and $\beta$.
    \item  Finally, the person walks all of the remaining path $P$.
\end{enumerate}
It can be checked that this maintains a maximum distance of 4 during all of the phases.
\end{proof}

\begin{lem}
\label{lem:lower-bound-necessary}
    If $\frechet(P,Q) \leq 4$, then $(U,V)$ is a yes-instance of OV.
\end{lem}

\begin{proof}
Consider a traversal of $P$ and $Q$ where at all times, the person and the dog are within distance $4$ of one another. 
At some point of time, the dog is on $\alpha^\star$. This implies the person is on a vertex $A$ right before some orange gadget $G_u$ corresponding to some $u \in U$. 
Now the person cannot move on his own, so the dog has to move.
The dog may walk through the blue gadgets for some time on his own, but at some point of time $t_0$ the person has also to move away from $A$ into $G_u$. 
More precisely, let $t_0$ be the smallest time such that at this time the position of the person is the vertex 0 in box $b_1(u)$. 
Since $z'$ is at distance greater than $4$ from $0$ and $1$, at time $t_0$ the dog is not on $z'$. Furthermore, due to the properties of $\beta^\star$, at time $t_0$ the dog has also not gone further than $z'$. 
This means the dog is standing inside some blue vector gadget at time $t_0$.
Since $\dist(I', 0) = 5$, the dog is not standing on $I'$, and hence is adjacent to some box.
But now due to Lemma~\ref{lem:gadget-necessary-single-step} the person and the dog have to move \emph{with the same speed} at least until the person has fully traversed all $d$ boxes of $G_u$.
We note that we need to watch out for the edge case where the person has not fully traversed $G_u$ yet, but the dog is in the very last box $b'_d(w^{(m)})$. In this edge case, Lemma~\ref{lem:gadget-necessary-single-step} is not applicable because $\nnext(b'_d(w^{(m)}))$ is not defined. But this is no problem. In this case the dog cannot move, since the only vertices within distance 4 of $z'$ are $A,B$ but the person is currently inside $G_u$. 
But the person cannot get to the next box as well, since the person is standing in a box $b_i(u)$ with $i \leq d - 1$, and $\dist(1',I) = 5$. So in this case they are stuck, a contradiction to $\frechet(P,Q) \leq 4$.
In all other cases that are not the edge case, i.e.\ the case where the person stands in a box $b_i(u)$ with $i \leq d-1$ and the dog stands in a box $b' \neq b'_d(w^{(m)})$, both $\nnext(b_i(u))$ and $\nnext(b')$ are defined, and so \cref{lem:gadget-necessary-single-step} is applicable.
In summary, this implies that $u$ is orthogonal to a consecutive substring of length $d$ of the string $\str(V)$. Due to the preprocessing (Lemma~\ref{lem:preprocessing}), this actually means that $(U, V)$ is a yes-instance.
\end{proof}

\section{Approximating the discrete Fréchet distance under any $L_c$ metric}
\label{app:euclidean}

We extend our techniques for curves in the plane under the $L_1$ metric to curves in the plane under any $L_c$ metric, for $c \geq 1$.
Let $\gamma > 0$ be a universal constant, and fix $\eps \in (0,1)$.
We consider $\mathbb{R}^2$ under the $L_c$ metric, and two polygonal curves $P$ and $Q$ with $n$ and $m$ vertices, respectively.
In full generality, we only require that all edges in $P$ have length at most $\gamma$ and that $Q$ is an $(\eps, \delta)$-curve. 
Denote their discrete Fréchet distance under the $L_c$ metric by $D_F(P, Q)$. We show how to compute a $(1+\eps)$-approximation of $D_F(P, Q)$ in $\tilde{O}(\frac{\sqrt{\delta}}{\eps^{1.5}} (n+m)^{1.5})$ time.
In particular, we design a  \emph{large-scale} $(1+\eps)$-\emph{decision algorithm}:

\begin{definition}
For any value $\Delta$ and $\eps > 0$, a large-scale $(1+\eps)$-decision algorithm outputs \textsc{Yes} or \textsc{No} under the following conditions:
    \begin{itemize}
        \item If $D_F(P, Q) \leq \Delta$, it must output \textsc{Yes}.
        \item If $D_F(P, Q) > (1+\eps)\Delta$, it must output \textsc{No}.
        \item Otherwise, its output may be arbitrary. 
        \item In the special case, when $\Delta \leq 1$ it must give the exact answer.
    \end{itemize}
\end{definition}

We use the result of~\cite{DriemelHW12} to compute a $(1+\eps)$-approximation of $D_F(P, Q)$, given a large-scale $(1+\eps)$-decision algorithm:

\begin{lemma}[Lemma 3.7 in~\cite{DriemelHW12}]
\label{lem:approx}
    Suppose that there exists a large-scale $(1+\eps)$-decision algorithm that, for any $\Delta$, runs in $T$ time.
    If $D_F(P, Q)$ lies in a given interval of width $N$, then this algorithm can compute a $(1+\eps)$-approximation of $D_F(P, Q)$ in $O(T \log \frac{N}{\eps})$ time. 
\end{lemma}

\subparagraph*{Algorithm overview.}
Our approach mirrors that of Section~\ref{sec:L1}.
Given $P$ and $Q$, we use Lemma~\ref{lem:properties_L1} to compute a tiling $\T$ where edges have width $t \in \Theta(\frac{\eps^{1.5} \cdot (n+m)}{\sqrt{\delta n}})$.
This partitions $P$ and $Q$ into at most $O( (n+m)/t)$ induced subcurves $S(P, \T)$ and $S(Q, \T)$. 
The product set $S(P, \T) \times S(Q, \T)$ induces a grid over $M_\Delta$ with $O((n+m)/t)$ rows and columns.

In Section~\ref{sec:L1}, we distinguished between pairs of faces of $\T$ that are aligned and not aligned.
There, we used that any pair of subcurves $(P', Q') \in S(P, \T) \times S(Q, \T)$ in unaligned faces is well-separated, giving rise to an efficient way of updating the wave-front over their corresponding submatrices.
In this section, we take another approach.
Here, we distinguish between pairs of faces of $\T$ based on the distance between them.

\begin{definition}
    We say that two faces $F, F'$ in a tiling $\T$ with edge length $t$ are \emph{far apart} whenever each point in $F$ is at least distance $2^{1+1/c} \cdot \frac{t}{\eps}$ (under the $L_c$ norm) away from any point in $F'$.
    The faces are \emph{close together} otherwise.
\end{definition}

This different distinction becomes useful only now, since we allow approximations.
Critically, for two subcurves $(P', Q') \in S(P, \T) \times S(Q, \T)$ in faces that are far apart, we may approximate pairwise distances by simply the distance between the two faces, thus treating the pairwise distances as constant.
In the remainder, we let $d(\cdot, \cdot)$ denote the distance function under the $L_c$ metric.

\begin{lemma}
    \label{lem:approximate}
    Let $F$ and $F'$ be two faces of $\T$ that are far apart.
    Let $D$ denote the minimum distance (under the $L_c$ metric) between a point in $F$ and a point in $F'$.
    For any two points $p \in F$ and $q \in F'$, we have $D \leq d(p, q) \leq (1+\eps)D$.
\end{lemma}
\begin{proof}
    Naturally, $D \leq d(p, q)$.
    Given that $F$ (resp. $F'$) is a square with sides of length $t$, the distance between any two points in $F$ (resp. $F'$) is at most $2^{1/c} \cdot t$.
    Thus $d(p, q) \leq D + 2^{1+1/c} \cdot t$.
    We have that $D \geq 2^{1+1/c} \cdot \frac{t}{\eps}$, since $F$ and $F'$ are far apart.
    Hence
    \[
        D \leq d(p, q) \leq D + 2^{1+1/c} \cdot t \leq (1+\eps)D. \qedhere
    \]
\end{proof}

\paragraph*{The wave-front algorithm.}
We adapt the wave-front propagation technique from Section~\ref{sec:upper_bound}.
We define the \emph{wave-front} $W$ as a monotone chain in $M_\Delta$ that separates $(1, 1)$ from $(n, m)$. 
For each grid point $(x, y)$ in $W$, we store the Boolean value:
\begin{itemize}
    \item \textsc{True} if there exists a path $F$ from $(1, 1)$ to $(x, y)$ where  $\forall (i, j) \in F$, $d(p, q) \leq \Delta$. 
    \item \textsc{False} if there exists no  path $F$ from $(1, 1)$ to $(x, y)$ where $\forall (i, j) \in F$, $d(p, q) \leq (1 + \eps)\Delta$. 
    \item Or an arbitrary Boolean value otherwise. 
\end{itemize}

A wave-front update selects any $(P[i:i'], Q[j:j']) \in S(P, \T) \times S(Q, \T)$ where the bottom facet of $M_\Delta[i:i', j:j']$ is incident to $W$ and it adds its top facet to $W$.
To perform a wave-front update, we distinguish between whether $P[i:i']$ and $Q[j:j']$ lie in close together or far apart faces.

Let $F_P$ and $F_Q$ be faces of $\T$ containing $P[i:i']$ and $Q[j:j']$, respectively.
We first discuss updating the wave-front when $F_P$ and $F_Q$ are far apart.
In this case, we compute the minimum distance $D$ between points in the two faces in constant time.
If $D \leq \Delta$, then by Lemma~\ref{lem:approximate} we have that $d(p, q) \leq (1+\eps)\Delta$ for all $p \in P[i:i']$ and $q \in Q[j:j']$.
We then treat the entire submatrix $M_\Delta[i:i', j:j']$ as containing only $0$-entries, making updating the wave-front a trivial $O(|P[i:i']| + |Q[j:j']|)$-time procedure.
If $D > \Delta$, all entries in $M_\Delta[i:i', j:j']$ have the value $1$, as $D$ is the minimum distance between points in $F_P$ and $F_Q$.
This again makes updating the wave-front a trivial $O(|P[i:i']| + |Q[j:j']|)$-time procedure.

Next we discuss updating the wave-front when $F_P$ and $F_Q$ are close together.
In this case, we upper bound the total number of switching cells $M_\delta[i, j]$ where $p_i$ and $q_j$ lie in faces of $\T$ that are close together.

\begin{lemma}
    \label{lem:intersection_count_L_2}
    There are at most $O(\frac{\delta n t}{\eps^3})$ pairs $(p_i, q_j)$ such that $M_\Delta[i, j]$ is a switching cell and $p_i$ and $q_j$ lie in faces of $\T$ that are close together.
    We can compute these pairs in $O(\frac{\delta n t}{\eps^3} \log m)$ time. 
\end{lemma}
\begin{proof}
    We preprocess $Q$ by snapping its vertices to a square grid $G$ of edge length $\eps^{-1}$, forming $Q'$.
    Each vertex in $Q'$ corresponds to $O(\delta)$ original vertices since $Q$ is an $(\eps, \delta)$-curve.
    We store $Q'$ in a spatial data structure (e.g., a lexicographically sorted balanced binary tree), where a query point $z$ returns all $q_j \in Q$ that have $z$ as their corresponding point in $Q'$ in $O(\delta + \log n)$ time.

    Recall that all edges of $Q$ have length at most $\gamma \in O(1)$. 
    Fix a vertex $p_i$ and let $F_P$ be a face of $\T$ that contains $p_i$.       
    Let $B_i(\Delta)$ be the metric circle of radius $\Delta$ centred at $p_i$. Let $\Gamma$ denote a ball of radius $2 \gamma$ and let $B_i^*(\Delta)$ be the metric annulus that is the Minkowski-sum $\Gamma \oplus B_i(\Delta)$. 
    Observe that an edge $(q_j, q_{j+1})$ of $Q$ intersects $B_i(\Delta)$ only if one of its corresponding vertices $q_j'$ and $q_{j+1}'$ of $Q'$ lies in $B_i^*(\Delta)$. 

    We prove that there are at most $O(\frac{\delta t}{\eps^3})$ vertices $q_j$ of $Q$ that lie in a face $F_Q$ that is close to $F_P$, such that $M_\Delta[i, j]$ is a switching cell.
    Suppose there is at least one such vertex $q_j$, otherwise the claim trivially holds.
    From the definition of switching cell, we have that $d(p_i, q_j) \leq \Delta$ and that either $d(p_i, q_{j-1}) > \Delta$ or $d(p_i, q_{j+1}) > \Delta$.
    The distance between $p_i$ and $q_j$ is at most $O(\frac{t}{\eps})$.
    Using that all edges of $Q$ have a maximum length of $\gamma \in O(1)$, we obtain that $\Delta \in O(\frac{t}{\eps})$.

    There are $O(\frac{\Delta}{\eps^2}) = O(\frac{t}{\eps^3})$ vertices of the grid $G$ that lie in $B_i^*(\Delta)$.
    We iterate over these vertices and for each we perform a membership query on $Q'$.
    This returns $O(\frac{\delta t}{\eps^3})$ vertices of $Q'$.
    For each reported vertex $q_j'$ we test if $M_\Delta[i, j]$ is a switching cell in constant time.
    Thus, for any $p_i \in P$, we compute all switching cells $M_\Delta[i, j]$ that correspond to the lemma statement in $O(\frac{\delta t}{\eps^3})$ time.
\end{proof}

\begin{restatable}{theorem}{upperbound_L2}
    \label{thm:upperbound_L2}
    Fix a universal constant $\gamma \in O(1)$.
    For any $\eps > 0$, if $P$ has edge length at most $\gamma$ and $Q$ is an $(\eps, \delta)$-curve, then we can compute a $(1+\eps)$-approximation of $D_F(P, Q)$, under the $L_c$ metric for any $c \geq 1$, in $\tilde{O}(\frac{\sqrt{\delta}}{\eps^{1.5}} (n+m)^{1.5})$ time.  
\end{restatable}
\begin{proof}
    The proof is analogous to that of Theorem~\ref{thm:upperbound}.
    Fix $\Delta \geq 0$ and let $t \geq 1$ to be set later.
    We construct a large-scale $(1+\eps)$-decision algorithm.
    We then note that $D_F(P, Q) \in [d_1, d_2]$ where $d_1$ is the minimum distance between any points in their bounding box and $d_2$ is the maximum distance. The theorem then follows from Lemma~\ref{lem:approximate}.

    We construct a large-scale $(1+\eps)$-decision algorithm.
In the special case where $\Delta \leq 1$, we simply compute all switching cells of the entire free-space matrix $M_\Delta$ by brute force. The proof of Lemma~\ref{lem:intersection_count_L_2} 
guarantees that there are at most $O(\eps^{-2} \gamma \delta)$ switching cells per column which can be computed in $O(\eps^{-2} \gamma \delta \log m)$ time per column.
Thus, in the special case where $\Delta \leq 1$, we can decide whether $D_F(P, Q) \leq \Delta$ in $O(m \log m + n \cdot \eps^{-2} \gamma \delta \log m)$ time.

    What remains is the more general case where $\Delta$ is at least $1$.
    Given $T$, $P$, and $Q$, $\Delta$ and $\eps$, we apply Lemma~\ref{lem:properties_L1} to compute in $O( (n + m)t)$ time our tiling $\T$ and the $O( (n+m)/t)$ induced  subcurves $S(P, \T)$ and $S(Q, \T)$.
    The product set $S(P, \T) \times S(Q, \T)$ induces a grid over $M_\Delta$ with  $O( (n+m)/t)$ rows and columns.
    Thus:

    \begin{equation}
    \label{eq:upper_bound-L2}
    \sum_{\textnormal{submatrices } M' \textnormal{ of } M_\Delta} \hspace{-1cm} |\textnormal{facets}(M')| = \sum_{(P[i:i'], Q[j:j']) \in S(P, \T) \times S(Q, \T)} \hspace{-0.8cm} O(|P[i:i']| + |Q[j:j']|)  \hspace{1cm} \in O( (n+m)^2/ t).
     \end{equation}

We iteratively perform wave-front updates with the input $(P[i:i'], Q[j:j'])$. 
Let $F_P$ and $F_Q$ denote the corresponding two faces of $\T$. 
If $F_P$ and $F_Q$ are far apart then we update the wave-front in $O(|P[i:i']| + |Q[j:j']|)$ time, as we treat the submatrix corresponding to the subcurves as containing either only $0$-entries or only $1$-entries.
The total time spent in this case equals:

\[
    \sum_{ (P[i:i'], Q[j:j']) \textnormal{ where } (F_P, F_Q) \textnormal{ are far apart  } } \hspace{-2.5cm} O(|P[i:i']| + |Q[j:j']|) \hspace{1cm} \in O( (n+m)^2/ t)
\]

    If $F_P$ and $F_Q$ are close together then we apply Lemma~\ref{lem:intersection_count_L_2} instead.
    The total time spent is:

   \[
    \sum_{P[i:i'], Q[j:j']) \textnormal{ where } (F_P, F_Q) \textnormal{ are close together}} \hspace{-2.5cm}  O(|P[i:i']| + |Q[j:j']|)  + \textnormal{switching cells}(M_\Delta[i:i', j:j']) )
\]
The first term is upper bounded by  $O(\frac{(n+m)^{2}}{t})$ by Equation~\ref{eq:upper_bound-L2}.
The second term is upper bounded by $O( \frac{\delta n t}{\eps^3} \log m)$ by combining Lemma~\ref{lem:brute-force} and \ref{lem:intersection_count_L_2}.
It follows that if we choose $t \in \Theta(\frac{\eps^{1.5} \cdot (n+m)}{\sqrt{\delta n}})$ the total running time of our large-scale $(1+\eps)$-decision procedure is upper bounded by $O(\frac{\sqrt{\delta}}{\eps^{1.5}}(n+m)^{1.5} \log m)$.
\end{proof}

\section{A continuous variant of $(\eps, \delta)$-curves.}
\label{app:continuous}

For the sake of completeness, we briefly mention a natural continuous generalisation of $(\eps, \delta)$-curves. We observe that, trivially, our algorithms can also approximate the continuous Fréchet distance for these continuous $(\eps, \delta)$-curves.

\begin{definition}
    Fix a universal constant $\gamma \in O(1)$. Let $P$ be a curve in the plane.
    Denote by $B_\eps$ a ball of radius $\eps$ and by $E_\eps$ the set of regions obtained by taking for every edge $e$ of $P$ the Minkowsky-sum  $B_\eps \oplus e$.
    We say that $P$ is a continuous $(\eps, \delta)$-curve if all edges have length at most $\gamma$ and all points of the plane intersect at most $\delta$ regions in $E_\eps$.
\end{definition}

We note that, trivially, if we require our curves to be continuous $(\eps, \delta)$-curves then our approximation algorithms for the discrete Fréchet distance also apply to the continuous Fréchet distance, at a factor $\eps^{-1.5}$ overhead. 

Specifically, let $P$ be a curve whose edges have length at most $\gamma$, and such that $P$ has $n$ vertices. 
Let $Q$ be a continuous $(\eps, \delta)$-curve with $m$ vertices.
Let $D(P, Q)$ denote their continuous Fréchet distance and let $P'$ (and $Q'$) be the ordered sequence of points obtained by placing points at distance $\frac{\eps}{8}$ along $P$ and $Q$.

\begin{observation}
    The discrete Fréchet distance $D_F(P', Q')$ is a $(1+\frac{\eps}{2})$-approximation of the continuous Fréchet distance $D(P, Q)$.
\end{observation}

\begin{observation}
    The sequence $P'$ has $O(\frac{n}{\eps})$ vertices and the sequence $Q'$ is an $(\eps, \delta)$-curve with $O(\frac{m}{\eps})$ vertices.
\end{observation}

Thus, we may apply Theorem~\ref{thm:upperbound_L2} to $P'$ and $Q'$ to approximate the continuous Fréchet distance $D(P, Q)$:

\begin{theorem}
        Fix a universal constant $\gamma \in O(1)$.
    For any $\eps > 0$, let $P$ has edge length at most $\gamma$ and $Q$ be a continuous $(\eps, \delta)$-curve. 
    If we denote by $D(P, Q)$ the continuous Fréchet distance, under the $L_c$ metric for any $c \geq 1$, then we can compute a value $F$ with $D(P, Q) \in [F - \eps, F + \eps]$ in  $\tilde{O}(\frac{\sqrt{\delta}}{\eps^{3}} (n+m)^{1.5})$ time.  
\end{theorem}

\noindent
This theorem offers a $(1+\eps)$-approximation when $D(P, Q)$ is large enough. Since this section is only meant to illustrate the continuous generalization of $(\eps, \delta)$-curves for the sake completeness, we consider the special case where $D(P, Q)$ is small to be out of scope.

\section{Comparing Well-behaved curve classes}
\label{app:realistic}

We discuss how the class of $(\eps, \delta)$-curves relates to existing classes of ``realistic'' curves. 
We focus on the following established curve classes: $c$-packed~\cite{DriemelHW12, bringmann17cpacked}, $\phi$-low-density~\cite{DriemelHW12}, $\kappa$-straight~\cite{DriemelHW12}, and $\kappa$-bounded curves~\cite{alt2004comparison,DriemelHW12}. 
These classes admit faster approximation algorithms for both the discrete and continuous Fréchet distance.

For completeness, we presented a continuous analogue of $(\eps, \delta)$-curves in Appendix~\ref{app:continuous}. 
However, we consider the discrete variant to be more natural, and thus we restrict our comparisons to the discrete setting.

A key distinction between $(\eps, \delta)$-curves and the existing realistic curve classes is that $(\eps, \delta)$-curves are defined in $\mathbb{R}^2$, whereas $c$-packed and $\kappa$-straight curves admit efficient $(1+\eps)$-approximation algorithms in $\mathbb{R}^d$ for constant $d$. 
For $\phi$-low-density and $\kappa$-bounded curves, approximation algorithms exist but with running time that scales exponentially in $d$.

Another key distinction is that, under the $L_1$ metric,  $(\eps, \delta)$-curves allow for the exact computation of the discrete Fréchet distance in subquadratic time. 

\subparagraph{$\phi$-low-density curves.}
We first consider the class of $\phi$-low-density curves~\cite{DriemelHW12, buchin_low_density_spanenr}.
A curve $P$ is said to be $\phi$-low-density if any ball of radius $r$ intersects at most $\phi$ edges of $P$ that are longer than $r$.
This input model resembles our definition of $(\eps, \delta)$-curves: any ball of radius $\eps$ intersects at most $\delta$ vertices of an $(\eps, \delta)$-curve.
In particular, setting $\phi = \delta$ allows both definitions to permit a vertex to coincide with up to $\phi = \delta$ other vertices.

There are, however, notable differences. In some respects, $(\eps, \delta)$-curves are more restrictive:
\begin{itemize}
    \item $(\eps, \delta)$-curves impose a global upper bound on edge lengths, independent of the input complexity.
    \item $\phi$-low-density curves allow for arbitrarily tight clusters of short edges.
\end{itemize}

Conversely, $(\eps, \delta)$-curves are less restrictive in other ways:
\begin{itemize}
    \item They only enforce local constraints within balls of radius $\eps$, making the condition inherently more local.
    \item The algorithmic dependency on $\phi$ in $(1+\eps)$-approximation algorithms is poor. The existing running times are only stated for constant $\phi$~\cite{DriemelHW12}.
\end{itemize}

\subparagraph{$c$-packed curves.}
The class of $c$-packed curves is perhaps the most extensively studied~\cite{DriemelHW12, 10.1145/3643683, DBLP:conf/compgeom/DriemelHR22, bringmann17cpacked, vanderhoog_et_al:approximate}.
A curve is $c$-packed if its length inside any ball is at most $c$ times the ball’s radius (note that overlapping subcurves are counted with multiplicity).
Any $c$-packed curve is also $2c$-low-density~\cite{DriemelHW12}, so the same comparison with $(\eps, \delta)$-curves applies.

However, $c$-packed curves are even more restrictive. For example, consider a tightly wound spiral of short segments. Such a curve can be $O(1)$-low-density and satisfy the definition of an $(\eps,1)$-curve for sufficiently small $\eps$, yet it is $\Theta(n)$-packed.

\subparagraph{$\kappa$-straight curves.}
A curve $P$ is $\kappa$-straight if the arc length between any two points on $P$ is at most $\kappa$ times their Euclidean distance.
Any $\kappa$-straight curve is also $\kappa$-packed. 
This class is strictly more restrictive than $c$-packed curves and, notably, excludes self-intersecting curves. 
This stands in stark contrast to $(\eps, \delta)$-curves, which allow for such intersections.

\subparagraph{$\kappa$-bounded curves.}
Finally, we consider $\kappa$-bounded curves.
A curve $P$ is $\kappa$-bounded if, for any two points $P(x)$ and $P(x')$ on the curve, the subcurve between them lies within the union of two balls of radius $\kappa \cdot d(P(x), P(x'))$, each centred at $P(x)$ and $P(x')$ respectively.

Although $\kappa$-bounded curves may intersect a given ball arbitrarily many times, they have a significant limitation: they do not permit self-intersections.
This again contrasts sharply with the $(\eps, \delta)$-curve model, which permits such behaviour.

\end{document}